\documentclass[journal, romanappendices, final]{IEEEtran}
\usepackage{latexsym,amsfonts,amssymb,amsthm, setspace, verbatim,cite,mathrsfs,graphicx,bm,stfloats, hyperref, tikz,xcolor,soul, flushend}
\usepackage[cmex10]{amsmath}
\newcommand{\un}{\underline}
\newcommand{\be}{\begin{equation}}
\newcommand{\ee}{\end{equation}}
\newcommand{\ben}{\begin{equation*}}
\newcommand{\een}{\end{equation*}}
\newcommand{\mc}{\mathcal}
\newcommand{\rda}{(\rho^2-D)\alpha}
\newcommand{\rdabar}{(\rho^2-D)\bar{\alpha}}
\newcommand{\rdabara}{\rho^2 \bar{\alpha} + D \alpha}
\newcommand{\e}{\epsilon}
\newcommand{\mcb}{\mathcal{B}_{M,L}}
\newcommand{\bfs}{\mathbf{S}}
\newcommand{\tbfs}{\tilde{\bfs}}
\newcommand{\bfsh}{\mathbf{\hat{S}}}

\newcommand{\abs}[1]{\lvert#1\rvert}
\newcommand{\norm}[1]{\lVert#1\rVert}
\newcommand{\expec}{\mathbb{E}}
\newcommand{\mbf}{\mathbf}
\newcommand{\mck}{\mathcal{K}}
\newcommand{\mckc}{\mathcal{K}^c}
\newtheorem{lem}{Lemma}

\newtheorem{defi}{Definition}
\newtheorem{thm}{Theorem}
\newtheorem{fact}{Fact}

\newtheorem{corr}{Corollary}

\begin{document}
\title{The Rate-Distortion Function and Excess-Distortion Exponent of  Sparse Regression Codes \\ with Optimal Encoding}
\author{Ramji Venkataramanan,~\IEEEmembership{Senior Member,~IEEE,}
and Sekhar Tatikonda,~\IEEEmembership{Senior Member,~IEEE}

\thanks{This work was partially supported by a Marie Curie Career Integration Grant (Grant Agreement Number 631489) and by NSF Grant CCF-1217023. This paper was presented in part at the 2014 IEEE International Symposium on Information Theory.}%
\thanks{R.~Venkataramanan is with the Department of Engineering, University of Cambridge, Cambridge CB2 1PZ, UK (e-mail: ramji.v@eng.cam.ac.uk).}%
\thanks{S. Tatikonda is with the Department of Statistics and Data Science, Yale University, New Haven CT 06511, USA (e-mail: sekhar.tatikonda@yale.edu).}
%
\vspace{-10pt}
}
\maketitle

\begin{abstract}
This paper studies the  performance of sparse regression codes  for lossy compression with the squared-error distortion criterion. In a sparse regression code,  codewords are linear combinations of subsets of columns of a design matrix. It is shown that with minimum-distance encoding, sparse regression codes achieve the Shannon rate-distortion function for i.i.d. Gaussian sources  $R^*(D)$ as well as the optimal excess-distortion exponent. This completes a previous result which showed that $R^*(D)$  and the optimal exponent were achievable for distortions below a certain threshold. The proof of the rate-distortion result is based on the second moment method, a popular technique to show that a non-negative random variable $X$ is strictly positive with high probability. In our context, $X$ is the number of codewords  within target  distortion $D$ of the source sequence. We first identify the reason behind the failure of the standard second moment method for certain distortions, and illustrate the different failure modes via a stylized example. We then use a refinement of the second moment method to show that $R^*(D)$ is achievable for all distortion values. Finally, the refinement technique is applied to Suen's correlation inequality to prove the achievability of the optimal Gaussian excess-distortion exponent. 
\end{abstract}


\section{Introduction}

\IEEEPARstart{D}{eveloping} practical codes for lossy compression at rates approaching Shannon's rate-distortion bound has long been an important goal in information theory. A practical compression code requires  a codebook with
low storage complexity  as well as encoding and decoding with low computational complexity.  Sparse Superposition Codes or Sparse Regression Codes (SPARCs) are a recent class of codes introduced by Barron and Joseph, originally for communcation over the AWGN channel \cite{AntonyML, AntonyFast}. They were subsequently used for lossy compression with the squared-error distortion criterion in \cite{KontSPARC,RVGaussianML, RVGaussFeasible}.  The codewords in a SPARC are linear combinations of columns of a design matrix $\mbf{A}$. The storage complexity of the code is proportional to the size of the matrix, which is polynomial in the block length $n$. A computationally efficient encoder for  compression with SPARCs  was proposed in   \cite{RVGaussFeasible} and shown to achieve rates approaching the Shannon rate-distortion function for i.i.d. Gaussian sources.

In this paper, we study the compression performance of SPARCs with the squared-error distortion criterion under optimal (minimum-distance) encoding. We show that for any ergodic source with variance $\sigma^2$, SPARCs with optimal encoding achieve a rate-distortion trade-off  given by $R^*(D) := \tfrac{1}{2} \log \tfrac{\sigma^2}{D}$. Note that $R^*(D)$ is the optimal rate-distortion function for an i.i.d. Gaussian source with variance $\sigma^2$.  The performance of SPARCs with optimal encoding was first studied in \cite{RVGaussianML}, where it was shown that for any distortion-level $D$, rates greater than
\be R_0(D) := \max \left\{ \frac{1}{2} \log  \frac{\sigma^2}{D}, \, \left(1 - \frac{D}{\sigma^2} \right)\right\} \label{eq:rsp_def} \ee
are achievable with the optimal Gaussian excess-distortion exponent. The  rate $R_0(D)$ in \eqref{eq:rsp_def} is equal to $R^*(D)$ when $\frac{D}{\sigma^2} \leq x^*$, but is strictly larger than $R^*(D)$ when $\frac{D}{\sigma^2} > x^*$, where $x^* \approx 0.203$; see Fig. \ref{fig:rd_ml_perf}. In this paper, we complete the result of \cite{RVGaussianML} by proving that sparse regression codes achieve the Gaussian rate-distortion function $R^*(D)$ for all distortions $D \in (0, \sigma^2)$. We also show that these codes attain the optimal excess-distortion exponent for i.i.d. Gaussian sources at all rates.

\begin{figure}[t]
\centering
\includegraphics[width=3.25in]{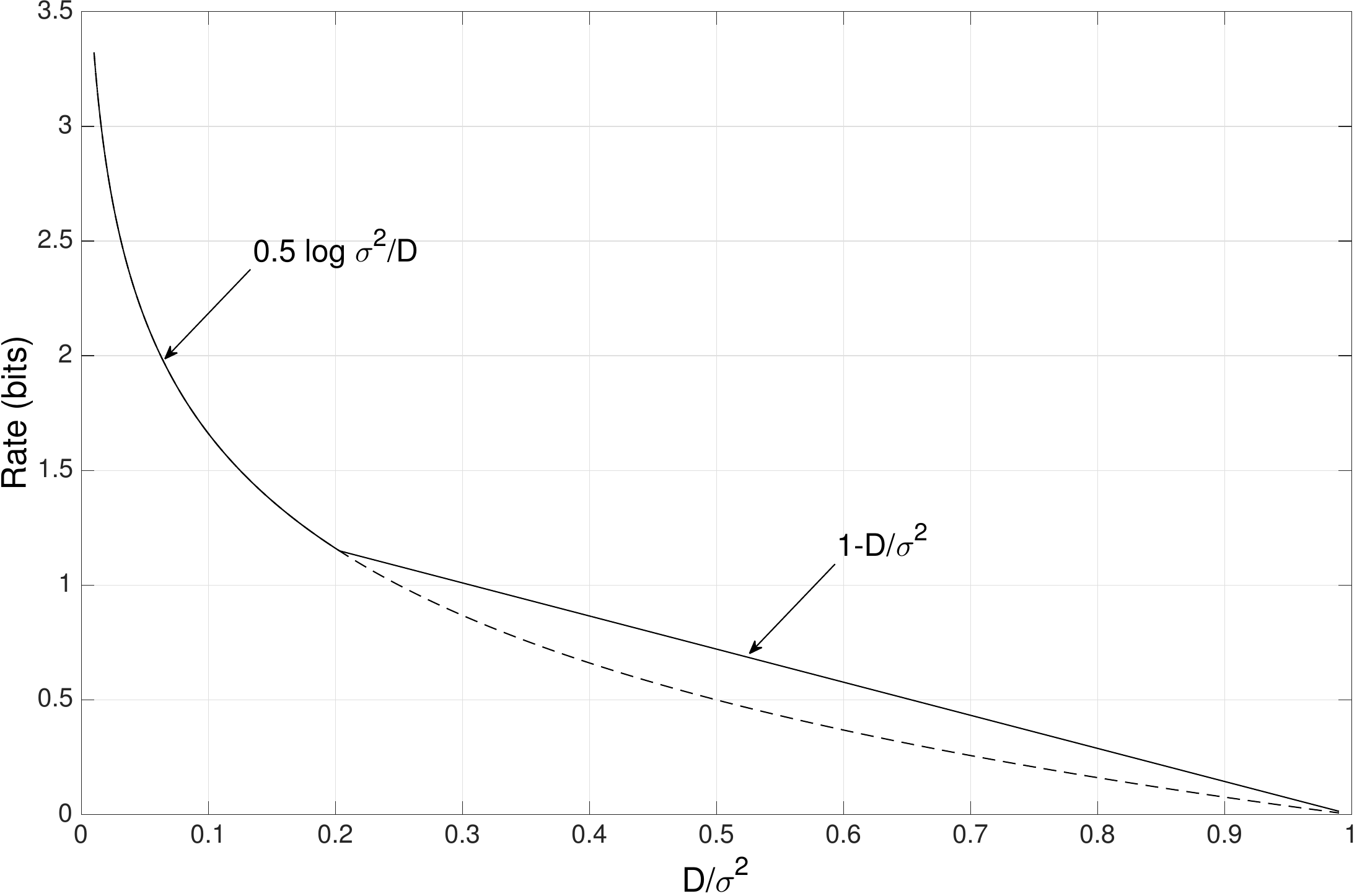}
\caption{\small{The solid line shows the previous achievable rate $R_0(D)$, given  in  \eqref{eq:rsp_def}. The rate-distortion function $R^*(D)$ is shown in dashed lines. It coincides with $R_0(D)$ for $D/\sigma^2 \leq x^*$, where $x^* \approx 0.203$.} }
\label{fig:rd_ml_perf}
\end{figure}

Though minimum-distance encoding is not practically feasible (indeed, the main motivation for sparse regression codes is that they enable low-complexity encoding and decoding), characterizing  the rate-distortion function and excess-distortion exponent under optimal encoding establishes a benchmark to compare the performance of various computationally efficient encoding schemes. Further, the results of this paper and \cite{RVGaussianML} together show that SPARCs retain the good covering properties of the i.i.d. Gaussian random codebook, while having a compact representation in terms of a matrix whose size is a low-order polynomial in the block length.

 Let us specify some notation before proceeding. Upper-case letters are used to denote random variables, and lower-case letters for their realizations. Bold-face letters are used to denote random vectors and matrices.  All vectors  have length $n$. The  source sequence  is $\bfs := (S_1, \ldots, S_n)$, and the reconstruction sequence is $\bfsh := (\hat{S}_1, \ldots, \hat{S}_n)$. $\norm{\mathbf{x}}$ denotes the $\ell_2$-norm of vector $\mathbf{x}$, and
$\abs{\mathbf{x}} =  \tfrac{\norm{\mathbf{x}}}{\sqrt{n}}$  is the normalized version.
$\mc{N}(\mu,\sigma^2)$ denotes the Gaussian distribution with mean $\mu$ and variance $\sigma^2$. Logarithms are with base $e$ and rate is measured in nats, unless otherwise mentioned. The notation $a_n \sim b_n$ means that $\lim_{n \to \infty} \tfrac{1}{n} \log a_n = \lim_{n \to \infty} \tfrac{1}{n} \log b_n $, and  w.h.p is used to abbreviate the phrase `with high probability'. We will use $\kappa, \kappa_1, \kappa_2$ to denote generic positive constants whose exact value is not needed. 

\subsection{SPARCs with Optimal Encoding} \label{sec:sparcdef}
A sparse regression code  is defined in terms of a design matrix $\mathbf{A}$ of dimension $n \times ML$ whose entries are i.i.d. $\mathcal{N}(0,1)$.
Here $n$ is the block length and $M$ and $L$ are integers whose values will be specified in terms of $n$ and the rate $R$.  As shown in Fig. \ref{fig:sparserd}, one can think of the matrix $\mathbf{A}$ as composed of $L$ sections with $M$ columns each. Each codeword is a linear combination of $L$ columns, with one column from each section.
Formally, a codeword can be expressed as  $\mathbf{A} \beta$, where $\beta$ is  an $ML \times 1$ vector $(\beta_1, \ldots, \beta_{ML})$ with the following property:  there is exactly one non-zero $\beta_i$ for  $1 \leq i \leq M$, one non-zero $\beta_i$ for $M+1 \leq i \leq 2M$, and so forth.  The non-zero values of $\beta$ are all set equal to $ \frac{c}{\sqrt{L}}$ where $c$ is a constant that will be specified later.  Denote the set of all $\beta$'s that satisfy this property by $\mcb$.

\begin{figure}[t]
\centering
\includegraphics[width=3.25in]{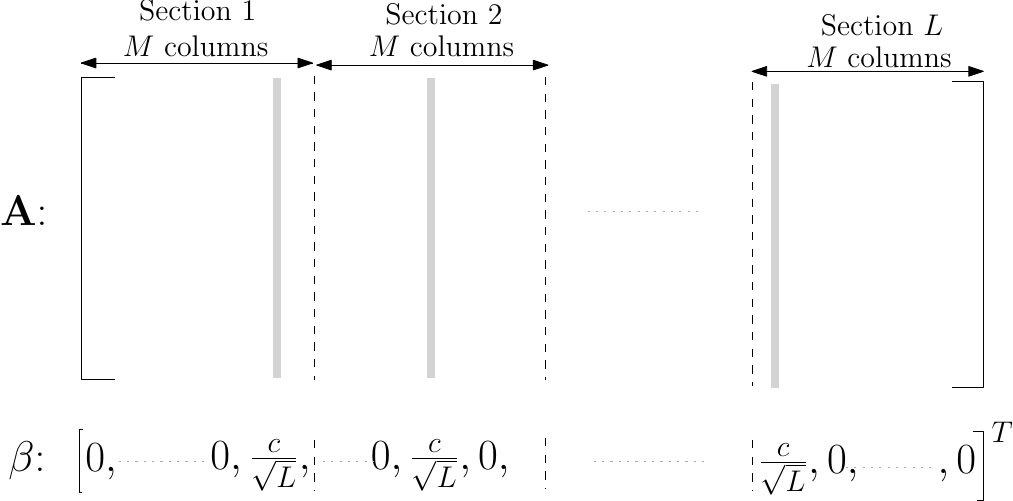}
\caption{\small{$\mathbf{A}$ is an $n \times ML$ matrix and $\beta$ is a $ML \times 1$ binary vector. The positions of the non-zeros in $\beta$  correspond to the gray columns of $\mathbf{A}$ which combine to form the codeword $\mathbf{A}\beta$.}}
\vspace{-5pt}
\label{fig:sparserd}
\end{figure}

\emph{Minimum-distance encoder}: This is defined by a mapping $g: \mathbb{R}^n \to \mcb$. Given the source sequence $\bfs$, the encoder determines the $\beta$ that produces the codeword closest in Euclidean distance, i.e.,
 \[ g(\bfs) = \underset{\beta \in \mcb}{\operatorname{argmin}} \ \norm{\bfs - \mathbf{A}\beta}.\]

\emph{Decoder}: This is a mapping $h: \mcb \to \mathbb{R}^n$. On receiving ${\beta} \in \mcb$ from the encoder, the decoder produces reconstruction $h(\beta) = \mathbf{A}\beta$.

Since there are $M$ columns in each of the $L$ sections, the total number of codewords is $M^L$. To obtain a compression rate of $R$ nats/sample, we therefore need
\be
M^L = e^{nR}.
\label{eq:ml_nR}
\ee
 For our constructions, we choose $M=L^b$ for some $b>1$ so that \eqref{eq:ml_nR} implies
\be   L \log L = \frac{nR}{b}. \label{eq:rel_nL} \ee
Thus $L$ is $\Theta\left( n/\log n \right)$, and the number of columns $ML$ in the dictionary $\mathbf{A}$ is  $\Theta\left(\left(n/\log n \right)^{b+1}\right)$, a {polynomial}  in $n$.

\subsection{Overview of our Approach}
To show that a rate $R$ can be achieved at distortion-level $D$, we need to show that with high probability at least one of  the $e^{nR}$ choices for $\beta$  satisfies
\be \abs{ \bfs  - \mbf{A} \beta }^2 \leq D. \label{eq:within_D}\ee
If $\beta$ satisfies \eqref{eq:within_D}, we call it  a \emph{solution}.

Denoting the number of solutions by $X$, the goal is to show that $X > 0$ with high probability when $R > R^*(D)$. Note that $X$ can be expressed as the sum of $e^{nR}$ indicator random variables, where the $i$th indicator is $1$ if $\beta(i)$ is a solution and zero otherwise, for $1 \leq i \leq e^{nR}$.
Analyzing the probability $P(X > 0)$ is challenging  because these indicator random variables are \emph{dependent}:  codewords $\mbf{A}\beta(1)$ and $\mbf{A}\beta(2)$ will be dependent if $\beta(1)$ and $\beta(2)$ share common non-zero terms. To handle the dependence,  we use the second moment method (second MoM), a technique commonly used to prove existence (`achievability') results in random graphs and random constraint satisfaction problems \cite{AlonSpBook}. In the setting of lossy compression, the second MoM was used in \cite{WainManeva10} to obtain the rate-distortion function of  LDGM codes for binary symmetric  sources with Hamming distortion.

For any non-negative random variable $X$,  the second MoM\cite{JansonBook} bounds the probability of the event $X > 0$ from below as\footnote{The inequality \eqref{eq:2nd_mom} follows from the Cauchy-Schwarz inequality $ (\expec[X Y])^2 \leq \expec X^2 \, \expec Y^2$ by substituting $Y=\mathbf{1}_{\{X>0\}}$.}
\be
P(X > 0) \geq \frac{(\expec X)^2}{\expec [ X^2]}.
\label{eq:2nd_mom}
\ee

Therefore   the second MoM succeeds if we can show that $(\expec X)^2/\expec [ X^2] \, \to \, 1$ as $ n \to \infty$.  It was shown in  \cite{RVGaussianML} that the second MoM succeeds for $R > R_0(D)$, where $R_0(D)$ is defined in \eqref{eq:rsp_def}. In contrast, for $R^*(D)< R < R_0(D)$  it was found that $(\expec X)^2/\expec [ X^2] \, \to \, 0$, so the second MoM fails. From this result in \cite{RVGaussianML}, it is not  clear whether the gap from $R^*(D)$ is due to an inherent weakness of the sparse regression codebook, or if it is just a limitation of the second MoM as a proof technique. In this paper, we demonstrate that it is the latter, and refine the second MoM to prove that all rates greater than
$R^*(D)$ are achievable.

Our refinement of the second MoM is inspired by  the  work of Coja-Oghlan and Zdeborov\'{a} \cite{CojaZdeb12} on finding sharp thresholds for  two-coloring of random hypergraphs. The high-level idea is as follows.  The key ratio $(\expec X)^2/\expec [ X^2]$ can be expressed as $(\expec X)/\expec [ X(\beta)]$, where $X(\beta)$ denotes the total number of solutions conditioned on the event that a given $\beta$ is a solution. (Recall that $\beta$ is a solution if $\abs{ \bfs  - \mbf{A} \beta }^2 \leq D$.) Thus when the  second MoM fails, i.e. the ratio goes to zero, we have a situation where the expected number of solutions is much smaller than the expected number of solutions \emph{conditioned} on the event that $\beta$ is a solution. This happens because for any $\mbf{S}$, there are atypical realizations of the design matrix that yield a very large number of solutions.  The total probability of these matrices is small enough that $\expec X$ in not significantly affected by these realizations. However, conditioning on $\beta$ being a solution increases the probability that the realized design matrix is one that yields an unusually large number of solutions. At low rates, the conditional probability of the design matrix being atypical is large enough to make $\expec [ X(\beta)] \gg \expec X$, causing the second MoM to fail.\footnote{This is similar to the inspection paradox in renewal processes.}

The key to rectifying the second MoM failure is to show that $X(\beta) \approx \expec X$ {with high probability} \emph{although} $\expec [ X(\beta)]  \gg \expec X$. We then apply the second MoM to count just the `good' solutions, i.e., solutions $\beta$ for which $X(\beta) \approx \expec X$. This succeeds,  letting us conclude that $X > 0$ with high probability.

\subsection{Related Work} As mentioned above, the second moment method was used in  \cite{WainManeva10} to analyze the rate-distortion function of LDGM codes for binary symmetric sources with Hamming distortion.
The idea of applying  the second MoM to a random variable that counts just the `good' solutions was recently used to obtain improved thresholds for problems such as random hypergraph 2-coloring \cite{CojaZdeb12}, $k$-colorability of random graphs \cite{coja13kcol},  and random $k$-SAT
\cite{coja13kSAT}. However, the key step of showing that a given solution is `good' with high probability depends heavily on the geometry of the problem being considered. This step requires identifying a specific property  of the random object being considered (e.g., SPARC design  matrix, hypergraph, or boolean formula)  that leads to a very large number of solutions  in atypical realizations of the object. For example, in SPARC compression, the atypical realizations are design matrices with columns that are unusually well-aligned with the source sequence to be compressed; in random hypergraph $2$-coloring, the atypical realizations  are  hypergraphs with an edge structure that allows an unusually large number of vertices to take on either color \cite{CojaZdeb12}.

It is interesting to contrast the analysis of SPARC lossy compression with that of SPARC AWGN channel coding in \cite{AntonyML}.  The dependence structure of the SPARC codewords makes the analysis  challenging in both problems, but the techniques required to analyze SPARC channel coding are very different from those used here for the excess distortion analysis. In the channel coding case, the authors use  a modified union bound together with a novel bounding technique for the probability of pairwise error events  \cite[Lemmas 3,4]{AntonyML} to establish that the error probability decays exponentially for all rates smaller than the channel capacity. In contrast, we use a refinement of the second moment method for the rate-distortion function, and Suen's correlation inequality to obtain the excess-distortion exponent.

Beyond the excess-distortion exponent, the \emph{dispersion}  is another quantity of interest in a lossy compression problem \cite{IngberKochman,KostinaV12}. For a fixed excess-distortion probability, the dispersion specifies how fast the rate can approach the rate-distortion function with growing block length. It was shown that for discrete memoryless and i.i.d. Gaussian sources, the optimal dispersion was equal to the inverse of the second derivative of the excess-distortion exponent. Given that SPARCs attain the optimal excess-distortion exponent, it would be interesting to explore if they also achieve the optimal dispersion for i.i.d. Gaussian sources with squared-error distortion. 

The rest of the paper is organized as follows. The main results, specifying the rate-distortion function and the excess-distortion expoenent of SPARCs, are stated in Section \ref{sec:main_result}. In Section \ref{sec:2mom_fail},  we set up the proof and show why the second MoM fails for $R < (1- \tfrac{D}{\rho^2})$. As the proofs of the main theorems are technical, we motivate the main ideas with a stylized example in Section \ref{subsec:toy_example}.  The main results are proved in Section \ref{sec:proof}, with the proof of the main technical lemma given in Section \ref{sec:proof_lem}.

\section{Main Results} \label{sec:main_result}

 The probability of excess distortion at distortion-level $D$ of a rate-distortion code $\mathcal{C}_n$ with block length $n$ and encoder and decoder mappings $g,h$  is
\be P_{e}(\mathcal{C}_n, D) = P\left(\abs{\bfs - h(g(\bfs))}^2 > D\right). \label{eq:pedef} \ee
For a SPARC generated as described in Section \ref{sec:sparcdef}, the probability measure in \eqref{eq:pedef} is with respect to the random source sequence $\bfs$ and the random design matrix $\mbf{A}$.

\subsection{Rate-Distortion Trade-off of SPARC} \label{subsec:sparc_rd}
\begin{defi}
A rate $R$ is achievable at distortion level $D$  if there exists a sequence of SPARCs $\{\mathcal{C}_n\}_{n=1,2,\ldots}$ such that
$\lim_{n \to \infty} P_{e}(\mathcal{C}_n, D) =0$
where for all $n$, $\mathcal{C}_n$ is a rate $R$ code defined by an $n \times L_n M_n$ design matrix whose parameter $L_n$ satisfies \eqref{eq:rel_nL} with a fixed $b$ and $M_n=L_n^b$.
\end{defi}

\begin{thm}
Let $\bfs$ be drawn from an ergodic source with mean $0$ and variance $\sigma^2$. For $D \in (0, \sigma^2)$,  let $R^*(D) = \tfrac{1}{2} \log \tfrac{\sigma^2}{D}$. Fix $R > R^*(D)$ and $b > b_{min}(\frac{\sigma^2}{D})$, where  
\be 
\begin{split}
&  b_{min} \left(x \right) =  \\
 &  \frac{20 R \, x^4}{ \left(1 + \frac{1}{x} \right)^2\left(1 -\frac{1}{x}\right) \left[ -1 + \left( 1 +  
\frac{2 \sqrt{x}}
{(x -1)} \left(R -\frac{1}{2}(1-\frac{1}{x}) \right)\right)^{1/2}\right]^2}
\end{split}
\label{eq:bmin_def}
\ee
for $1 < x \leq e^{2R}$. Then there exists a sequence of rate $R$  SPARCs $\{\mc{C}_n\}_{n=1,2 \ldots}$  for which $\lim_{n \to \infty} P_{e}(\mathcal{C}_n, D) =0$, where $\mathcal{C}_n$ is defined by an $n  \times L_n M_n$ design matrix, with $M_n =L_n^b$ and $L_n$ determined by \eqref{eq:rel_nL}.
\label{thm:ml_result}
\end{thm}
\emph{Remark}:  Though the theorem  is valid for all $D \in (0, \sigma^2)$, it is most relevant for the case $\frac{D}{\sigma^2} > x^*$, where $x^* \approx 0.203$ is the solution to the equation 
\[ (1-x) + \frac{1}{2} \log x =0.  \]
For $\frac{D}{\sigma^2}  \leq x^*$, \cite[Theorem $1$]{RVGaussianML} already guarantees that  the optimal rate-distortion function can be achieved, with a smaller value of $b$ than that required by the theorem above.

\subsection{Excess-distortion exponent of SPARC} \label{subsec:err_exp}

The excess-distortion exponent at distortion-level $D$ of a sequence of rate $R$ codes $\{\mathcal{C}_n\}_{n=1,2,\ldots}$ is given by
\be r(D,R) = - \limsup_{n \to \infty} \frac{1}{n} \log  P_{e}(\mathcal{C}_n,D), \ee
where $P_{e}(\mathcal{C}_n,D)$ is defined in \eqref{eq:pedef}. The optimal excess-distortion exponent for a rate-distortion pair $(R,D)$ is the supremum of the excess-distortion exponents over all sequences of codes with rate  $R$, at distortion-level $D$.

The optimal excess-distortion exponent for discrete memoryless sources was obtained by Marton \cite{MartonRD74}, and the result was extended to memoryless Gaussian  sources by Ihara and Kubo \cite{IharaKubo00}. 

\begin{fact}\cite{IharaKubo00}
\label{fact:ihara}
For an i.i.d. Gaussian source distributed as $\mathcal{N}(0, \sigma^2)$ and  squared-error distortion criterion, the optimal excess-distortion exponent at rate $R$ and distortion-level $D$ is
\be
r^*(D,R) = \left\{
\begin{array}{ll}
 \frac{1}{2} \left( \frac{a^2}{\sigma^2} - 1 - \log \frac{a^2}{\sigma^2} \right) & \quad R> R^*(D) \\
 0 & \quad R \leq R^*(D)
\end{array}
\right.
\label{eq:opt_exp}
\ee
where  $a^2 = D e^{2R}$.
\end{fact}
For $R> R^*(D)$, the exponent in \eqref{eq:opt_exp} is the  Kullback-Leibler divergence between two zero-mean Gaussians, distributed as $\mc{N}(0,a^2)$ and $\mc{N}(0,\sigma^2)$, respectively.

The next theorem characterizes the excess-distortion exponent performance of SPARCs.
\begin{thm}
Let $\bfs$ be drawn from an ergodic source with mean zero and variance $\sigma^2$. Let $D \in (0, \sigma^2)$, $R > \frac{1}{2} \log \frac{\sigma^2}{D}$,  and  
$\gamma^2 \in (\sigma^2, De^{2R})$. Let
\be b> \max \left\{ 2, \  \frac{7}{5} b_{min}\left( {\gamma^2}/{D} \right) \right\}, \label{eq:bmin_exp} \ee 
where $b_{min}(.)$ is defined in \eqref{eq:bmin_def}.
Then there exists a sequence of rate $R$ SPARCs $\{C_n\}_{n=1,2 \ldots}$, where $\mathcal{C}_n$ is defined by an $n  \times L_n M_n$ design matrix with $M_n =L_n^b$ and $L_n$ determined by \eqref{eq:rel_nL}, whose probability of excess distortion at distortion-level $D$ can be bounded as follows for all sufficiently large $n$. 
\be
P_e(\mc{C}_n, D)  \leq P(\abs{\bfs}^2 \geq \gamma^2) + \exp\left( - \kappa n^{1+c} \right),
\ee
where  $\kappa, c$ are strictly positive universal constants.
\label{thm:err_exp}
\end{thm}

\begin{corr}
Let $\bfs$ be drawn from an i.i.d. Gaussian source with mean zero and variance $\sigma^2$. Fix rate $R > \frac{1}{2} \log \frac{\sigma^2}{D}$, and let $a^2=De^{2R}$. Fix any $\e \in (0, a^2 -\sigma^2)$, and 
\be
b > \max \left\{2, \, \frac{7}{5} b_{min}\left(\frac{a^2- \e}{D} \right) \right \}.
\label{eq:b_a2e}
\ee 
There exists a sequence of rate $R$ SPARCs with  parameter $b$ that achieves the excess-distortion exponent 
\[  \frac{1}{2} \left( \frac{a^2 - \e}{\sigma^2} -1 - \log \frac{a^2-\e}{\sigma^2} \right). \]
Consequently, the supremum of excess-distortion exponents  achievable by SPARCs for i.i.d. Gaussian sources sources is equal to the optimal one, given by \eqref{eq:opt_exp}.
\label{corr:err_exp}
\end{corr}
\begin{IEEEproof}
From Theorem \ref{thm:err_exp}, we know that for any $\e \in (0, a^2-\sigma^2)$, there exists a sequence of rate $R$ SPARCs $\{C_n\}$ for which
\be
P_e(\mc{C}_n, D)  \leq P(\abs{\bfs}^2 \geq a^2 - \e) \left(1 + \frac{\exp(-\kappa n^{1+c})}{P(\abs{\bfs}^2 \geq a^2-\e)}\right)
\label{eq:pe_rew}
\ee
for sufficiently large $n$, as long as the  parameter $b$ satisfies \eqref{eq:b_a2e}.
For $\bfs$ that is  i.i.d. $\mc{N}(0, \sigma^2)$, Cram{\'e}r's large deviation theorem \cite{Den2008LD} yields
\be 
\begin{split} 
& \lim_{n \to \infty} - \frac{1}{n} \log P(\abs{\bfs}^2 \geq a^2 - \e ) \\
& =  \frac{1}{2} \left( \frac{a^2- \e}{\sigma^2} -1 - \log \frac{a^2-\e}{\sigma^2} \right)
\end{split}
 \label{eq:gaussian_ld} 
\ee
for $(a^2 - \e) > \sigma^2$. Thus $P(\abs{\bfs}^2 \geq a^2- \e)$ decays exponentially with $n$; in comparison $\exp( - \kappa n^{1+c})$ decays \emph{faster} than exponentially with $n$. Therefore, from \eqref{eq:pe_rew},  the excess-distortion exponent satisfies
\be
\begin{split}
& \liminf_{ n \to \infty} \, \frac{-1}{n} \log P_e(\mc{C}_n, D)  \\
& \geq  \liminf_{ n \to \infty}  \frac{-1}{n}\Big[ \log P(\abs{\bfs}^2 \geq a^2 - \e)   \\
& \qquad \qquad \left. + \log \left(1 + \frac{\exp(- \kappa n^{1+c})}{P(\abs{\bfs}^2 \geq a^2 - \e)} \right) \right] \\
& = \frac{1}{2} \left( \frac{a^2 - \e}{\sigma^2} -1 - \log \frac{a^2-\e}{\sigma^2} \right).
\end{split}
\ee
Since $\e >0$ can be chosen arbitrarily small, the supremum of all achievable excess-distortion exponents is  \[ \frac{1}{2} \left( \frac{a^2}{\sigma^2} - 1 - \log \frac{a^2}{\sigma^2} \right),\] which is optimal from Fact \ref{fact:ihara}.
\end{IEEEproof}

We remark that the function $b_{min}(x)$ is increasing in $x$. Therefore \eqref{eq:b_a2e} implies that larger values of the design parameter $b$ are required to achieve  excess-distortion exponents closer to the optimal value (i.e., smaller values of $\e$ in Corollary \ref{corr:err_exp}).


\section{Inadequacy of the Direct Second MoM} \label{sec:2mom_fail}

\subsection{First steps of the proof} \label{subsec:proof_begin}

Fix a rate $R > R^*(D)$, and $b$ greater than the minimum value specified by the theorem. Note that $D e^{2R} > \sigma^2$ since $R > \tfrac{1}{2} \log \tfrac{\sigma^2}{D}$. Let $\gamma^2$ be any number such that $\sigma^2 < \gamma^2 < D e^{2R}$.

\emph{Code Construction}:
 For each block length $n$, pick $L$ as specified by \eqref{eq:rel_nL} and $M=L^b$.
Construct an $n \times ML$ design matrix $\mathbf{A}$ with entries drawn i.i.d. $\mathcal{N}(0, 1)$. The codebook consists of all vectors $\mathbf{A} \beta$ such that $\beta \in \mcb$. The non-zero entries of $\beta$ are all set equal to a value specified below.

\emph{Encoding and Decoding}: If the source sequence $\bfs$ is such that $\abs{\bfs}^2 \geq \gamma^2$, then the encoder declares an error. If $\abs{\bfs}^2 \leq D$, then $\bfs$ can be trivially compressed to within distortion $D$ using the all-zero codeword. The addition of this extra codeword to the codebook affects the rate in a negligible way.

 If $\abs{\bfs}^2 \in (D, \gamma^2)$, then $\bfs$ is compressed in two steps. First, quantize $\abs{\bfs}^2$ with an $n$-level uniform scalar quantizer $Q(.)$ with  support in the interval $(D, \gamma^2]$. For input $x\in(D, \gamma^2]$, if 
 \[  x \in \left( D+ \frac{(\gamma^2-D) (i-1)}{n},  \ D+ \frac{(\gamma^2-D) i}{n} \right], \]
 for   $i \in \{1, \ldots, n\}$, then  the quantizer output is
\ben
Q(x) = D+  \frac{(\gamma^2-D) (i-\tfrac{1}{2})}{n}. 
\een
Conveying the scalar quantization index to the decoder (with an additional $\log n$ nats)  allows us to adjust the codebook variance according to the norm of the observed source sequence.\footnote{The scalar quantization step is only included to simplify the analysis. In fact, we could use the same codebook variance $(\gamma^2 - D)$ for all $\bfs$ that satisfy 
 $\abs{\bfs}^2 \leq (\gamma^2-D)$, but this would make the forthcoming large deviations analysis quite cumbersome.} The non-zero entries of $\beta$ are each set to $\sqrt{(Q(\abs{\bfs}^2)-D)/L}$ so that each SPARC codeword has variance $(Q(\abs{\bfs}^2)-D)$.
 Define  a ``quantized-norm" version of $\bfs$ as
 \be
 \tilde{\bfs} := \sqrt{\frac{Q(\abs{\mbf{S}}^2)}{\abs{\mbf{S}}^2}} \, \mbf{S}.
 \ee
 Note that $\abs{\tilde{\bfs}}^2= Q(\abs{\mbf{S}}^2)$. We  use the SPARC to compress $\tilde{\bfs}$. The encoder finds 
\[ \hat{\beta} := \underset{\beta \in \mcb}{\operatorname{argmin}} \ \norm{\tilde{\bfs} - \mathbf{A}\beta}^2. \]
The decoder receives $\hat{\beta}$ and reconstructs $\bfsh = \mathbf{A} \hat{\beta}$. Note that  for block length $n$, the total number of bits transmitted by encoder is 
$\log n + L \log M$, yielding an overall rate of $R + \tfrac{\log n}{n}$ nats/sample.

\emph{Error Analysis}: For $\bfs$ such that $\abs{\bfs}^2 \in (D, \gamma^2)$, the overall distortion can be bounded as 
\be
\begin{split}
& \abs{\bfs - \mathbf{A} \hat{\beta}}^2 = \abs{\bfs - \tilde{\bfs} + \tilde{\bfs} - \mathbf{A} \hat{\beta}}^2  \\
&  \leq   \abs{\bfs - \tilde{\bfs}}^2 +  2 \abs{\bfs - \tilde{\bfs}}\abs{\tilde{\bfs} - \mbf{A} \hat{\beta}} +  \abs{\tilde{\bfs} -  \mbf{A} \hat{\beta}} ^2 \\
& \leq \frac{\kappa_1}{n^2} + \frac{\kappa_2 \abs{\tilde{\bfs} - \mbf{A}\hat{\beta}}}{n} +  \abs{\tilde{\bfs} - \mbf{A} \hat{\beta}} ^2
\end{split}
\label{eq:dist_overall}
\ee
for some positive constants $\kappa_1, \kappa_2$. The last inequality holds because the step-size of the scalar quantizer is 
$\frac{(\gamma^2-D)}{n}$, and $\abs{\bfs}^2 \in (D, \gamma^2)$.

 Let $\mc{E}(\tilde{\bfs})$ be the event that the minimum of $\abs{\tilde{\bfs} - \mathbf{A}\beta}^2$ over $\beta \in \mcb$ is greater than $D$. The encoder declares an error if  $\mc{E}(\tilde{\bfs})$ occurs. If $\mc{E}(\tilde{\bfs})$  \emph{does not} occur, 
the overall distortion in \eqref{eq:dist_overall} can be bounded as 
\be
\abs{\bfs - \mathbf{A} \hat{\beta}}^2 \leq D + \frac{\kappa}{n},
\label{eq:good_dist_bnd}
\ee
 for some positive constant $\kappa$. The overall rate (including that of the scalar quantizer) is $R + \frac{\log n}{n}$.
 
Denoting the probability of excess distortion for this random code by $P_{e,n}$, we have
\be
\begin{split}
P_{e,n}  &  \leq P(\abs{\bfs}^2 \geq \gamma^2) +  \max_{\rho^2 \in (D, \gamma^2)} P(\mc{E}(\tilde{\bfs}) \mid \abs{\tilde{\bfs}}^2 = \rho^2).
\end{split}
\label{eq:err_bound}
\ee

As $\gamma^2 > \sigma^2$, the ergodicity of the source guarantees that
\be
\lim_{n \to \infty} P(\abs{\bfs}^2 \geq \gamma^2) = 0.
\label{eq:s_erg}
\ee
To bound the second term in \eqref{eq:err_bound}, without loss of generality we can assume  that the source sequence \[ \tbfs = (\rho, \ldots, \rho). \] This is because the codebook distribution is rotationally invariant, due to the  i.i.d.  $\mc{N}(0,1)$ design matrix 
$\mbf{A}$. For any $\beta$, the entries of $\mbf{A} \beta(i)$ i.i.d. $\mc{N}(0,\rho^2-D)$.
We enumerate the codewords as $\mbf{A} \beta(i)$, where $\beta(i) \in \mcb$ for $i=1,\ldots, e^{nR}$.  

Define the indicator random variables
\be
U_i(\tbfs) = \left\{
\begin{array}{ll}
1 & \text{ if } \abs{\mbf{A} \beta(i) - \tbfs}^2 \leq D,\\
0 & \text{ otherwise}.
\end{array} \right.
\label{eq:ui_def}
\ee
We can then write
\be
P( \mc{E}(\tbfs)  ) = P\left(\sum_{i=1}^{e^{nR}} U_i(\tbfs) =0 \right).
\label{eq:sum_ui}
\ee
For a fixed $\tbfs$,  the $U_i(\tbfs)$'s are dependent. To see this,  consider codewords $\bfsh(i),\bfsh(j)$  corresponding to the  vectors ${\beta}(i), {\beta}(j) \in \mcb$, respectively. Recall that a vector in $\mcb$ is uniquely defined by the position of the non-zero value in each of its $L$ sections.   If ${\beta}(i)$ and ${\beta}(j)$ overlap in $r$ of their non-zero positions, then the column sums forming codewords $\bfsh(i)$ and $\bfsh(j)$ will share $r$ common terms, and consequently $U_i(\tbfs)$ and $U_j(\tbfs)$ will be dependent.

For brevity, we henceforth denote $U_i(\tilde{\bfs})$ by just $U_i$.  Applying the second MoM with 
\[ X :=\sum_{i=1}^{e^{nR}} U_i, \] we have from \eqref{eq:2nd_mom}
\be
P(X > 0) \geq \frac{(\expec X)^2}{\expec [ X^2]} \stackrel{(a)}{=} \frac{\expec X}{\expec [ X | \, U_1 =1]}
\label{eq:2mom_cond_expec}
\ee
where $(a)$ is obtained by expressing $\expec [ X^2]$ as follows.
\be
\begin{split}
& \expec[X^2] = \expec\left[X \sum_{i=1}^{e^{nR}} U_i \right] = \sum_{i=1}^{e^{nR}} \expec[X U_i] \\
&  = \sum_{i=1}^{e^{nR}} P(U_i =1) \expec [X | U_i =1] \\
& = \expec X \cdot  \expec [X | \, U_1 =1].
\end{split}
\label{eq:x2expand}
\ee
The last equality in \eqref{eq:x2expand} holds because  $\expec X = \sum_{i=1}^{e^{nR}} P(U_i =1)$, and due to the symmetry of the code construction. As $\expec[X^2] \geq (\expec X)^2$,  \eqref{eq:2mom_cond_expec} implies that $\expec [ X | \, U_1 =1] \geq \expec X$. Therefore, to show that $X>0$ w.h.p, we need
\be
\frac{\expec [X | \, U_1 =1]}{\expec X} \to 1 \  \text{ as }  \ n \to \infty.
\ee
\subsection{ $\expec X$ versus $\expec [X | \, U_1 =1]$} \label{subsec:ex_v_condex}

To compute $\expec X$, we derive a general lemma specifying the probability that a randomly chosen i.i.d $\mc{N}(0, y)$ codeword  is within distortion $z$ of a source sequence $\bfs$ with $\abs{\bfs}^2=x$. This lemma will be used in other parts of the proof as well. 
\begin{lem}
Let $\bfs$ be a vector with $\abs{\bfs}^2=x$. Let $\bfsh$ be an  i.i.d. $\mc{N}(0, y)$ random vector  that is  independent of $\bfs$. Then for $x,y,z >0$ and sufficiently large $n$, we have
\be 
\frac{\kappa}{\sqrt{n}} e^{-n f(x,y,z)} \leq P \left( \abs{\bfsh - \bfs}^2 \leq  z   \right) \leq  e^{-n f(x,y,z)},
\label{eq:Pfxyz}
\ee
where $\kappa$ is a universal positive constant and for $x,y,z >0$,  the large-deviation rate function $f$ is 
\be
f(x, y, z) = \left\{
\begin{array}{l l}
 \frac{x+z}{2y} - \frac{xz}{A y} - \frac{A}{4y} -\frac{1}{2} \ln\frac{A}{2x} & \text{ if }  z \leq x+y,  \\
 0 & \text{ otherwise}, \\
\end{array}
\right.
\label{eq:fdef}
\ee
and
\be
A = \sqrt{y^2 + 4 x z } - y.
\label{eq:Adef}
\ee
\label{lem:gen_sc}
\end{lem}
\begin{IEEEproof}
We have
\be 
\begin{split}
&P \left( \abs{\bfsh - \bfs}^2 \leq  z  \right) = P\left(\frac{1}{n}\sum_{k=1}^n (\hat{S}_k- S_k)^2 \leq  z  \right) \\
& =  P\left(\frac{1}{n}\sum_{k=1}^n (\hat{S}_k- \sqrt{x} )^2 \leq  z  \right),
\end{split}
\ee
where the last equality is due to the rotational invariance of the distribution of $\bfsh$, i.e.,  $\bfsh$ has the same joint distribution as $\mathbf{O}\bfsh$ for any orthogonal (rotation) matrix $\mathbf{O}$. In particular,  we choose $\mathbf{O}$ to be the matrix that rotates $\bfs$ to the vector $(\sqrt{x}, \ldots, \sqrt{x})$, and note that $ \abs{\bfsh - \bfs}^2 =  \abs{\mbf{O} \bfsh -  \mbf{O}\bfs}^2$. Then, using the strong version of Cram{\'e}r's large deviation theorem due to Bahadur and Rao \cite{Den2008LD, BahadurR60},  we have
\be
 \frac{\kappa}{\sqrt{n}} e^{-nI(x,y,z)}  \leq P\left(\frac{1}{n}\sum_{k=1}^n (\hat{S}_k- x)^2 \leq  z  \right) \leq e^{-nI(x,y,z)},
\ee
where the large-deviation rate function $I$ is given by 
\be
I(x,y,z) = \sup_{\lambda \geq 0 } \left\{ \lambda z  - \log \expec e^{\lambda(\hat{S} -\sqrt{x})^2} \right\}.
\label{eq:Idef}
\ee
The expectation on the RHS of \eqref{eq:Idef} is computed with $\hat{S}  \sim \mc{N}(0,y)$. Using standard calculations, we obtain
\be
 \log \expec e^{\lambda(\hat{S} -\sqrt{x})^2} = \frac{\lambda x}{1-2y\lambda} - \frac{1}{2} \log (1-2y \lambda), \qquad  \lambda <2y.
 \label{eq:logmomgen}
\ee
Substituting the expression in \eqref{eq:logmomgen} in \eqref{eq:Idef} and maximizing over $\lambda \in [0, 2y)$ yields $I(x,y,z) =f(x,y,z)$, where $f$ is given by \eqref{eq:fdef}.
\end{IEEEproof}
The expected number of solutions is given by
\be \expec X = e^{nR} P(U_1 =1) =  e^{nR} P\left( \abs{\mbf{A} \beta(1) - \tbfs}^2 \leq D \right).  \label{eq:EX} \ee
Since $\tbfs=(\rho, \rho, \ldots, \rho)$, and $\mbf{A} \beta(1)$ is i.i.d. $\mc{N}(0, \rho^2-D)$,  applying Lemma \ref{lem:gen_sc} we obtain the bounds
\be \frac{\kappa}{\sqrt{n}} e^{nR} e^{-n f(\rho^2, \rho^2-D,D)}  \leq \expec X   \leq e^{nR} e^{-n f(\rho^2, \rho^2-D,D)},  \label{eq:EXub} \ee
Note that \be f(\rho^2, \rho^2-D, D) = \frac{1}{2} \log \frac{\rho^2}{D}. \label{eq:fnice} \ee 

Next consider $ \expec [X | \, U_1 =1]$. If $\beta(i)$ and $\beta(j)$ overlap in $r$ of their non-zero positions, the column sums forming codewords $\bfsh(i)$ and $\bfsh(j)$ will share $r$ common terms. Therefore,
\be
\begin{split}
&  \expec [X | \, U_1 =1]   = \sum_{i=1}^{e^{nR}} P(U_i=1 | \, U_1 =1) \\
  & = \sum_{i=1}^{e^{nR}} \frac{P(U_i=1, \,  U_1 =1)}{P(U_1=1)}  \\
&  \stackrel{(a)}{=} \sum_{r=0}^{L} {L \choose r} (M-1)^{L-r} \frac{P(U_2=U_1=1 | \, \mc{F}_{12}(r))}{P(U_1=1)}
\end{split}
\label{eq:x|u1a}
\ee
where $\mc{F}_{12}(r)$ is the event that the codewords corresponding to ${U}_1$ and ${U}_2$ share $r$ common terms.  In \eqref{eq:x|u1a}, $(a)$ holds because for each codeword $\bfsh(i)$, there are a total of ${L \choose r} (M-1)^{L-r}$  codewords which share exactly $r$ common terms with $\bfsh(i)$, for $0\leq r \leq L$.  From \eqref{eq:x|u1a} and \eqref{eq:EX}, we obtain
\be
\begin{split}
 & \frac{\expec [X | \, U_1 =1]}{\expec X}  \\
 & = \sum_{r=0}^{L} {L \choose r} (M-1)^{L-r} \frac{P(U_2=U_1=1 | \, \mc{F}_{12}(r))} { e^{nR}\, (P(U_1=1))^2}  \\
 & \stackrel{(a)} \sim  1 +  \sum_{\alpha =  \frac{1}{L}, \ldots, \frac{L}{L}} {L \choose L\alpha} \frac{P(U_2=U_1=1 | \, \mc{F}_{12}(\alpha))}{ M^{L\alpha} \, (P(U_1=1))^2}  \\
& \stackrel{(b)}{=} 1 +  \sum_{\alpha =  \frac{1}{L}, \ldots, \frac{L}{L}} e^{n\Delta_\alpha}
\end{split}
\label{eq:alph_summands}
\ee
where $(a)$ is obtained by substituting $\alpha = \tfrac{r}{L}$ and $e^{nR}=M^L$. The notation $x_L\sim y_L$ means that $x_L/y_L \to 1$  as $L \to \infty$. The equality $(b)$ is from \cite[Appendix A]{RVGaussianML}, where it was also shown that
\be
\Delta_{\alpha} \leq \frac{\kappa}{L}+  \frac{R}{b} \min\{ \alpha, \, \bar{\alpha}, \, \tfrac{\log2}{\log L}\}  - h(\alpha)
\label{eq:delalph_bound}
\ee
where
\be
h(\alpha) := \alpha R - \frac{1}{2} \log \left( \frac{1+ \alpha}{1 - \alpha(1-\frac{2D}{\rho^2})} \right).
\ee
The inequality in \eqref{eq:delalph_bound} is asymptotically tight \cite{RVGaussianML}.
The term $e^{n \Delta_\alpha}$ in \eqref{eq:alph_summands} may be interpreted as follows. Conditioned on $\beta(1)$ being a solution, the expected number of solutions that share $\alpha L$ common terms with $\beta(1)$ is $\sim e^{n \Delta_\alpha} \expec X$. Recall that we require the left side of \eqref{eq:alph_summands} to tend to $1$ as $n \to \infty$. Therefore, we need  $\Delta_\alpha <0$ for $\alpha =\tfrac{1}{L}, \ldots, \tfrac{L}{L}$.  From \eqref{eq:delalph_bound}, we need $h(\alpha)$ to be positive in order to guarantee that $\Delta_\alpha <0$.  However, when $R < (1-\tfrac{D}{\rho^2})$,  it can be verified that  $h(\alpha)<0$  for $\alpha \in (0, \alpha^*)$ where $\alpha^* \in (0,1)$ is the solution to $h(\alpha)=0$.  Thus $\Delta_\alpha$ is \emph{positive} for $\alpha \in (0, \alpha^*)$ when $\frac{1}{2} \log \frac{\rho^2}{D} \leq R \leq (1-\tfrac{D}{\rho^2})$. Consequently,  \eqref{eq:alph_summands} implies that
\be
\frac{\expec [X | \, U_1 =1]}{\expec X} \sim  \sum_{\alpha} e^{n\Delta_\alpha} \to \infty \ \text{ as } \ n \to \infty, 
\label{eq:x|u1b}
\ee
and the second MoM fails.

\subsection{A Stylized Example} \label{subsec:toy_example}
Before describing how to rectify the second MoM failure in the SPARC setting, we present a simple example to give intuition about  the failure modes of the second MoM. The proofs  in the next two sections do not rely on the discussion here.

Consider a sequence of generic random structures (e.g., a sequence of random graphs or SPARC design matrices) denoted  by ${R}_n, \, n \geq 1$. Suppose that for each $n$, the realization of $R_n$ belongs to one of two categories: a category $\mathsf{C}_1$ structure  which  has which has $e^{n}$ solutions, or a category $\mathsf{C}_2$ structure which has $e^{2n}$ solutions. In the case of SPARC, a solution is a codeword that is within the target distortion. Let the probabilities of $R_n$ being of each category be
\be
P(R_n \in \mathsf{C}_1) =  1- e^{-np}, \quad P(R_n \in \mathsf{C}_2) = e^{-np},
\label{eq:prn}
\ee
where $p >0$ is a constant. Regardless of the realization, we note that $R_n$ always has at least $e^n$ solutions.  

We now examine whether the second MoM can guarantee the existence of a solution for this problem as $n \to \infty$. The number of solutions $X$ can be expressed as a sum of indicator random variables: 
\[ X =\sum_{i=1}^{N} U_i,
\]
where $U_i=1$ if configuration $i$ is a solution, and $N$ is the total number of configurations. (In the SPARC context, a configuration is a codeword.) We assume that the configurations are symmetric (as in the SPARC set-up), so that  each one has equal probability of being a solution, i.e., 
\be P(U_i =1 \mid R_n \in \mathsf{C}_1) = \frac{e^n}{N}, \ \   P(U_i =1 \mid R_n \in \mathsf{C}_2) = \frac{e^{2n}}{N}.  \label{eq: u_given_r}\ee

Due to symmetry, the second moment ratio can be expressed as 
\be
\frac{\expec X^2}{(\expec X)^2} = \frac{\expec[X \mid U_1=1]}{\expec X}  =  \frac{\expec[X \mid U_1=1]}{(1-e^{-np})e^n + e^{-np}e^{2n}}.
\label{eq:2mom_rat}
\ee
The conditional expectation in the numerator can be computed as follows.
\be
\begin{split}
& \expec[X | U_1=1]   =  P(R_n \in \mathsf{C}_1 \mid U_1=1) \expec[X | U_1=1, \mathsf{C}_1]  \\
& \qquad  \qquad +  P(R_n \in \mathsf{C}_2 \mid U_1=1) \expec[X | U_1=1, \mathsf{C}_2]  \\
& \stackrel{(a)}{=} \frac{(1-e^{-np}) (e^n/N)}{(1-e^{-np}) (e^n/N) + e^{-np} (e^{2n}/N) } \, e^{n} \\
&  +   \frac{e^{-np} (e^{2n}/N)}{(1-e^{-np}) (e^n/N) + e^{-np} (e^{2n}/N) } \, e^{2n} \\
& =  \frac{(1-e^{-np}) e^{2n} + e^{n(4-p)}}{(1-e^{-np}) e^n + e^{n(2-p)}},
\end{split}
\label{eq:cond_expec_calc0}
\ee
where $(a)$ is obtained by using Bayes' rule to compute $P(R_n \in \mathsf{C}_1 \mid U_1=1)$. The second MoM ratio in \eqref{eq:2mom_rat} therefore equals
\be
\frac{\expec X^2}{(\expec X)^2} = \frac{\expec[X \mid U_1=1]}{\expec X} =  \frac{(1-e^{-np}) e^{2n} + e^{n(4-p)}}{[(1-e^{-np}) e^n + e^{n(2-p)}]^2}.
\label{eq:key_2mom_ratio}
\ee
We examine the behavior of the ratio above as $n \to \infty$ for different values of $p$.

\textbf{Case $1$:} $p \geq 2$. The dominant term in both the numerator and the denominator of \eqref{eq:key_2mom_ratio} is 
$e^{2n}$, and we get
\be
\frac{\expec[X \mid U_1=1]}{\expec X} \to 1 \text{ as } n \to \infty,
\ee
and the second MoM succeeds.

\textbf{Case $2$:}  $1 < p  \leq 2$.  The dominant term in the numerator is $e^{n(4-p)}$, while the dominant term in the denominator is $e^{2n}$. Hence
\be \frac{\expec[X \mid U_1=1]}{\expec X} =\frac{e^{n(4-p)}}{e^{2n}}(1+o(1)) \sim e^{n(2-p)}  \stackrel{n \to \infty}{\longrightarrow} \infty.  \ee

\textbf{Case $3$:} $0< p \leq 1$.  The dominant term in the numerator is $e^{n(4-p)}$, while the dominant term in the denominator is $e^{n(4-2p)}$. Hence
\be \frac{\expec[X \mid U_1=1]}{\expec X}  =\frac{e^{n(4-p)}}{e^{n(4-2p)}}(1+o(1)) \sim  e^{np}  \stackrel{n \to \infty}{\longrightarrow} \infty.  \ee

Thus in both Case $2$ and Case $3$, the second MoM fails because the expected number of solutions conditioned on a solution $(U_1=1)$ is exponentially larger than the unconditional expected value. However, there is an important distinction between the two cases, which allows us to fix the failure of the second MoM in Case $2$ but not in Case $3$. 

Consider the conditional distribution of the number of solutions given $U_1=1$. From the calculation in \eqref{eq:cond_expec_calc0}, we have
\be
\begin{split}
& P(X=e^n \mid U_1=1) = P(R_n \in \mathsf{C}_1 \mid U_1=1) \\
& \quad = \frac{(1-e^{-np}) e^n}{(1-e^{-np}) e^n + e^{n(2-p)} }, \\
& P(X=e^{2n} \mid U_1=1) = P(R_n \in \mathsf{C}_2 \mid U_1=1)  \\
& \quad =  \frac{e^{n(2-p)}}{(1-e^{-np}) e^n + e^{n(2-p)} }.
\end{split}
\label{eq:cond_dist_formula}
\ee

 When $1 < p \leq 2$, the first term in the denominator of the RHS dominates, and the conditional distribution of $X$ is 
 \be
 \begin{split}
&  P(X=e^n \mid U_1=1)  = 1- e^{-n(p-1)}(1+o(1)), \\
 &  P(X=e^{2n} \mid U_1)  =  e^{-n(p-1)}(1+o(1)).
 \end{split}
 \ee
Thus the  conditional probability of a realization $R_n$ being category  $\mathsf{C}_1$ given $U_1=1$ is slightly smaller than the unconditional probability, which is $1-e^{-np}$. However, conditioned on $U_1=1$, a  realization $R_n$ is still extremely likely to have come from category $\mathsf{C}_1$, i.e., have $e^{n}$ solutions. Therefore, when $1 < p \leq 2$, conditioning on a solution \emph{does not} change the nature of the `typical' or `high-probability' realization. This makes it possible to fix the failure of the second MoM in this case. The idea is to define a new random variable $X'$ which counts the number of solutions coming from typical realizations, i.e.,  only category $\mathsf{C}_1$ structures. The second MoM is then applied to $X'$ to show that is strictly positive with high probability.

When $p <1$, conditioning on a solution completely changes the distribution of $X$.  The dominant term in the denominator of the RHS in \eqref{eq:cond_dist_formula} is $e^{n(2-p)}$, so the conditional distribution of $X$ is
 \be
 \begin{split}
 & P(X=e^n \mid U_1=1)  =  e^{-n(1-p)}(1+o(1)),   \\
 &  P(X=e^{2n} \mid U_1)  = 1- e^{-n(1-p)}(1+o(1)).
 \end{split}
 \ee
Thus, conditioned on a solution, a typical realization of $R_n$ belongs to category $\mathsf{C}_2$, i.e., has $e^{2n}$ solutions. On the other hand, if we draw from the {unconditional} distribution of $R_n$ in \eqref{eq:prn}, a typical realization has $e^n$ solutions. In this case, the second moment method cannot be fixed by counting only the solutions from realizations of category $\mathsf{C}_1$, because the total conditional probability of such realizations is very small. This is the analog of the ``condensation phase" that is found in problems such as random hypergraph coloring \cite{CojaZdeb12}. In this phase, although solutions may exist, even an enhanced second MoM does not prove their existence.  

Fortunately, there is no condensation phase in the SPARC compression problem. Despite the failure of the direct second MoM, we prove (Lemma \ref{lem:good}) that  conditioning on a solution does not significantly alter the total number of solutions for a very large fraction of design matrices. Analogous to Case $2$ above, we can apply the second MoM to a new random variable that counts only the solutions coming from typical realizations of the design matrix. This yields the desired result that solutions exist for all rates $R< R^*(D)$.

\section{Proofs of Main Results}  \label{sec:proof}

\subsection{Proof of Theorem \ref{thm:ml_result}} \label{subsec:thm1_proof}
The code parameters, encoding and decoding are as described in Section \ref{subsec:proof_begin}. We build on the proof set-up of Section \ref{subsec:ex_v_condex}.
Given that $\beta \in \mcb$ is a solution, for $\alpha =0, \tfrac{1}{L}, \ldots, \tfrac{L}{L}$ define  $X_{\alpha}(\beta)$ to  be the number of solutions that share $\alpha L$ non-zero terms with $\beta$. The \emph{total} number of solutions given that $\beta$ is a solution is
\begin{align}
X(\beta) & = \sum_{\alpha=0, \frac{1}{L}, \ldots, \frac{L}{L}} X_{\alpha}(\beta)
\end{align}
Using this notation,  we have  
\be
\begin{split}
& \frac{\expec [X | \, U_1 =1]}{\expec X} \stackrel{(a)}{=} \frac{\expec[X(\beta)]}{\expec X} \\
&  = \sum_{\alpha = 0, \frac{1}{L}, \ldots, \frac{L}{L}} \frac{\expec[X_{\alpha}(\beta)]}{\expec X} \ \stackrel{(b)}{\sim} \  1+  \sum_{\alpha = \frac{1}{L}, \ldots, \frac{L}{L}}  e^{n\Delta_\alpha},
\end{split}
\label{eq:exbeta}
\ee
where ($a$) holds because  the symmetry of the code construction allows us to condition on a generic $\beta \in \mcb$ being a solution; ($b$) follows from \eqref{eq:alph_summands}.  Note that $\expec[X_{\alpha}(\beta)]$ and $\expec[X(\beta)]$ are expectations   evaluated with the \emph{conditional} distribution over the space of design matrices given that $\beta$ is a solution.

The key ingredient in the proof is the following lemma, which shows that $X_{\alpha}(\beta)$ is much smaller than $\expec X$ w.h.p  $\forall \alpha \in \{ \frac{1}{L}, \ldots, \frac{L}{L}\}$. In particular, $X_\alpha(\beta) \ll \expec X$ \emph{even} for $\alpha$ for which
\[ \frac{\expec[X_{\alpha}(\beta)]}{\expec X}  \sim e^{n\Delta_\alpha} \to  \infty \ \text{ as } n \to \infty.   \]

\begin{lem}
Let $R> \tfrac{1}{2}\log \frac{\rho^2}{D}$. If $\beta \in \mc{B}_{M,L}$ is a solution, then for sufficiently large $L$ 
\be 
P\left( X_\alpha(\beta)  \leq   L^{-3/2} \, \expec X, \   \text{ for } \tfrac{1}{L}\leq \alpha \leq  \tfrac{L-1}{L} \right) \geq 1- \eta
\ee
where
\be \eta = L^{-2.5 \left(\frac{b}{b_{min}(\rho^2/D)} - 1 \right)}.  \label{eq:eta_def} \ee
The function $b_{min}(.)$ is defined in \eqref{eq:bmin_def}.
\label{lem:good}
\end{lem}

\proof The proof of the lemma is given in Section \ref{sec:proof_lem}.

The probability measure in  Lemma \ref{lem:good} is the conditional distribution on the space of design matrices $\mbf{A}$ given that $\beta$ is a solution.

\begin{defi}
For $\e >0$,  call a solution $\beta$ ``$\e$-good"  if
\be \sum_{\alpha = \frac{1}{L}, \ldots, \frac{L}{L}}X_\alpha (\beta) < \e \, {\expec X}. \label{eq:good_def} \ee
\label{def:good}
\end{defi}
Since we have fixed $\tbfs= (\rho, \ldots, \rho)$, whether a solution $\beta$ is $\e$-good or not is determined by the design matrix. Lemma \ref{lem:good} guarantees that w.h.p any solution $\beta$ will be $\e$-good, i.e., if $\beta$ is a solution, w.h.p the design matrix is such that  the number of solutions sharing any common terms with $\beta$ is less $\e \expec[X]$.

The key to proving Theorem \ref{thm:ml_result} is  to apply the second MoM only to $\e$-good solutions.
 Fix $\epsilon = L^{-0.5}$. For $i=1,\ldots, e^{nR}$, define the indicator random variables
 \be
V_i  = \left\{
\begin{array}{ll}
1 & \text{ if } \abs{\mbf{A} \beta(i) - \tbfs}^2 \leq D \ and  \  \beta(i) \text{ is $\e$-good} ,\\
0 & \text{ otherwise}.
\end{array} \right.
\label{eq:vi_def}
\ee
The number of $\e$-good solutions, denoted by $X_g$, is given by
\be
X_g = V_1 + V_2 + \ldots + V_{e^{nR}}.
\label{eq:xg_count}
\ee
We will apply the second MoM to $X_g$ to show that $P(X_g > 0) \to 1$ as $n \to \infty$.   We have
\be
P(X_g > 0) \geq \frac{(\expec X_g)^2}{\expec [ X_g^2]} \, = \, \frac{\expec X_g}{\expec [ X_g | \, V_1 =1]}
\label{eq:2mom_v}
\ee
where the second equality is obtained by writing  $\expec [ X_g^2] = (\expec X_g) \expec [ X_g | \, V_1 =1]$, similar to \eqref{eq:x2expand}.

\begin{lem}
a) $\expec X_g \geq (1- \eta) \expec X$, where $\eta$ is defined in \eqref{eq:eta_def}.

b) $\expec [ X_g | \, V_1 =1] \leq (1 + L^{-0.5}) \expec X$.
\label{lem:Xg}
\end{lem}
\begin{IEEEproof}
Due to the symmetry of the code construction, we have
\be
\begin{split}
\expec X_g &= e^{nR} P(V_1=1) \stackrel{(a)}{=} e^{nR} P(U_1=1) P(V_1=1|U_1=1) \\
&  = \expec X \cdot P(  \beta(1) \text{ is $\e$-good } \mid \beta(1) \text{ is a solution}  ).
\end{split}
\label{eq:pvu}
\ee
In \eqref{eq:pvu},  $(a)$ follows from the definitions of $V_i$ in \eqref{eq:vi_def} and $U_i$ in \eqref{eq:ui_def}.
Given that $\beta(1)$  is a solution, Lemma \ref{lem:good} shows that
\be
 \sum_{\alpha =  \frac{1}{L}, \ldots, \frac{L}{L}} X_{\alpha}(\beta(1)) < (\expec X )  L^{-0.5}.
\label{eq:b1good}
\ee
with probability at least $1 - \eta$. As $\e= L^{-0.5}$,   $\beta(1)$ is $\e$-good according to Definition \ref{def:good}  if \eqref{eq:b1good} is satisfied. Thus $\expec X_g$ in \eqref{eq:pvu} can be lower bounded as
\be
\expec X_g \geq  (1 - \eta) \expec X. 
\ee
For part (b),  first observe that the total number of solutions $X$ is an upper bound for the number of $\e$-good solutions $X_g$. Therefore
\be
\expec [ X_g | \, V_1 =1] \leq \expec [ X | \, V_1 =1]. 
\label{eq:exg_v1}
\ee
  Given that $\beta(1)$ is an $\e$-good solution, the expected number of solutions can be expressed as
  \be
  \begin{split}
 & \expec [ X | \, V_1 =1]  \\
  & =   \expec[X_0 (\beta(1)) \mid  V_1=1]   + \expec[\sum_{\alpha = \frac{1}{L}, \ldots, \frac{L}{L}} \hspace{-4pt} X_\alpha (\beta(1)) \mid  V_1=1].
  \end{split}
  \label{eq:ex_v1_split}
  \ee
  There are $(M-1)^L$ codewords that share \emph{no} common terms with $\beta(1)$. Each of these codewords is independent of $\beta(1)$, and thus independent of the event $V_1=1$.
  \be
  \begin{split}
  & \expec [X_0 (\beta(1)) \mid V_1=1]  = \expec[X_0(\beta(1))] \\
   & = (M-1)^L \, P(\abs{\tbfs - \mbf{A} \beta}^2 \leq D) \\
  & \leq M^L \, P(\abs{\tbfs - \mbf{A} \beta}^2 \leq D) \\
  & = \expec X.
  \end{split}
  \label{eq:ex0_bound}
  \ee
  Next,  note that conditioned on $\beta(1)$ being an $\e$-good solution (i.e., $V_1=1$),
  \be
  \sum_{\alpha = \frac{1}{L}, \ldots, \frac{L}{L}} \hspace{-4pt} X_\alpha (\beta(1)) < \e \, \expec X
  \label{eq:exalph_bound}
  \ee
  \emph{with certainty}. This follows from the definition of $\e$-good  in \eqref{eq:good_def}.  Using \eqref{eq:ex0_bound} and   \eqref{eq:exalph_bound} in \eqref{eq:ex_v1_split}, we conclude that
  \be
  \expec [ X | \, V_1 =1]  < (1 + \e)\expec X.
  \label{eq:x_v1_bnd}
  \ee
  Combining \eqref{eq:x_v1_bnd} with \eqref{eq:exg_v1} completes the proof of Lemma \ref{lem:Xg}.
  \end{IEEEproof}

Using Lemma \ref{lem:Xg} in \eqref{eq:2mom_v}, we obtain
\be
\begin{split}
& P(X_g > 0) \geq  \frac{\expec X_g}{\expec [ X_g | \, V_1 =1]} \geq \frac{(1-\eta)}{1+\e}  \\
& =  \frac{1-L^{-2.5 (\frac{b}{b_{min}(\rho^2/D)} - 1)}}{1+L^{-1/2}},
\end{split}
\ee
where the last equality is obtained by using the definition of $\eta$ in \eqref{eq:eta_def} and $\e=L^{-0.5}$. Hence the probability of the existence of at least one good solution goes to $1$ as $L \to \infty$. Thus we have shown that for any $\rho^2 \in (D, \gamma^2)$, the quantity \[ P(\mc{E}(\tilde{\bfs}) \mid \abs{\tilde{\bfs}}^2=\rho^2)\] in \eqref{eq:err_bound} tends to zero whenever  $R > \tfrac{1}{2} \log \tfrac{\rho^2}{D}$ and $b > b_{min}(\frac{\rho^2}{D})$. Combining this with \eqref{eq:good_dist_bnd}--\eqref{eq:s_erg},we conclude that that the probability that
\[
\abs{\bfs - \mathbf{A} \hat{\beta}}^2 \leq D + \frac{\kappa}{n}
\]
goes to one as $n \to \infty$.   As $\gamma^2 > \sigma^2$ can be chosen arbitrarily close to $\sigma^2$, the proof of Theorem \ref{thm:ml_result} is complete.

\subsection{Proof of Theorem \ref{thm:err_exp}}  \label{subsec:thm2_proof} 

The code construction is as described in Section \ref{subsec:proof_begin}, with the parameter $b$ now chosen to satisfy \eqref{eq:bmin_exp}.  Recall the definition of an 
$\e$-good solution in Definition \ref{def:good}. We follow the set-up of Section \ref{subsec:thm1_proof} and count the number of $\e$-good solutions, for an appropriately defined $\e$.  As before, we want an upper bound for the probability of the event $X_g=0$, where the number of $\e$-good solutions $X_g$ is defined in \eqref{eq:xg_count}.

Theorem \ref{thm:err_exp} is obtained  using Suen's correlation inequality to upper bound on the probability of the event $X_g=0$. Suen's inequality yields a sharper upper bound than the second MoM. We use it to prove that the probability of  $X_g=0$ decays \emph{super-exponentially} in $L$. In comparison, the second MoM only guarantees a polynomial decay. 

We begin with some definitions required for Suen's inequality.

\begin{defi}
[Dependency Graphs  \cite{JansonBook}] Let $\{V_i\}_{i \in \mc{I}}$ be a family of random variables (defined on a common probability space). A dependency graph for $\{V_i\}$ is any graph $\Gamma$ with vertex set $V(\Gamma)= \mc{I}$ whose set of edges satisfies the following property: if $A$ and $B$ are two disjoint subsets of  $\mc{I}$ such that there are no edges with one vertex in $A$ and the other in $B$, then the families $\{V_i\}_{i \in A}$ and $\{V_i\}_{i \in B}$
are independent.
\end{defi}

\begin{fact}
\cite[Example $1.5$, p.11]{JansonBook}
Suppose $\{Y_\alpha\}_{\alpha \in \mc{A}}$ is a family of independent random variables, and each $V_i, i\in \mc{I}$ is a function of the variables
$\{Y_\alpha\}_{\alpha \in A_i}$ for some subset $A_i \subseteq \mc{A}$. Then the graph with vertex set $\mc{I}$ and edge set
$\{ij : A_i \cap A_j \neq \emptyset\}$ is a dependency graph for  $\{U_i\}_{i \in \mc{I}}$.
\label{fact:depgraph_ex}
\end{fact}

In our setting, we fix $\e =L^{-3/2}$, let $V_i$ be the  indicator the random variable defined in \eqref{eq:vi_def}. Note that  $V_i$ is one if and only if  $\beta(i)$ is an $\e$-good solution. The set of codewords that share at least one common term with $\beta(i)$ are the ones that play a role in determining whether $\beta(i)$ is an $\e$-good solution or not. Hence, the graph $\Gamma$ with vertex set $V(\Gamma) = \{1,\ldots,e^{nR} \}$  and edge set $e(\Gamma)$ given by
\ben
\begin{split}  
&  \{ ij:  i\neq j \text{ and the codewords } \beta(i), \beta(j) \\
& \quad  \text { share at least one common term} \} 
 \end{split} \een 
 is a dependency graph for the family
$\{ V_i \}_{i=1}^{e^{nR}}$. This follows from Fact \ref{fact:depgraph_ex} by observing that: i) each  $V_i$ is a function of the columns of  $\mathbf{A}$ that define $\beta(i)$ and all other codewords that share at least one common term with $\beta(i)$; and ii) the columns of $\mathbf{A}$ are generated independently of one another.
\label{rem:depgraph}

 For a given codeword $\beta(i)$, there are ${L \choose r} (M-1)^{L-r}$ other codewords that have exactly $r$ terms in common with $\beta(i)$, for $0 \leq r \leq (L-1)$.   Therefore each vertex in the dependency graph for the family $\{V_i\}_{i=1}^{e^{nR}}$ is connected to
\[ \sum_{r=1}^{L-1} {L \choose r} (M-1)^{L-r} = M^L - 1 - (M-1)^L \]
other vertices.

\begin{fact}[Suen's Inequality \cite{JansonBook}] Let $V_i \sim \text{Bern}(p_i),  i\in \mc{I}$, be a finite family of Bernoulli random variables having a dependency graph $\Gamma$.  Write $i \sim j$ if $ij$ is an edge in $\Gamma$. Define
\ben
\begin{split}
\lambda  = \sum_{i\in \mc{I}} \expec V_i,
 \  \, \Delta = \frac{1}{2} \sum_{i \in \mc{I}} \sum_{j \sim i} \expec(V_i V_j),
\  \, \delta =  \max_{i\in \mc{I}} \sum_{k \sim i}  \expec V_k.
\end{split}
\een
Then
\be P\left(\sum_{i \in \mc{I}} V_i =0\right) \leq \exp\left(-\min \left\{ \frac{\lambda}{2}, \frac{\lambda}{6\delta}, \frac{\lambda^2}{8\Delta}  \right\} \right).
 \label{eq:suens_ineq} \ee
\end{fact}
We apply Suen's inequality  with the dependency graph specified above for $\{V_i\}_{i=1}^{e^{nR}}$ to compute an upper bound for $P(X_g =0)$, where $X_g= \sum_{i=1}^{e^{nR}} V_i$ is the total number of $\e$-good solutions for $\e =L^{-3/2}$. Note that the $\e$ chosen here  is smaller than the value of $L^{-1/2}$ used for Theorem \ref{thm:ml_result}. This smaller value is required to prove the super-exponential decay of the excess-distortion probability via Suen's inequality. We also need a stronger version of Lemma \ref{lem:good}.
\begin{lem}
\label{lem:good_strong}
Let $R> \tfrac{1}{2}\log \frac{\rho^2}{D}$.
If $\beta \in \mc{B}_{M,L}$ is a solution, then for sufficiently large $L$ 
\be 
P\left( X_\alpha(\beta)  \leq   L^{-5/2} \, \expec X, \   \text{ for } \tfrac{1}{L}\leq \alpha \leq  \tfrac{L}{L} \right) \geq 1- \xi
\ee
where
\be \xi = L^{-2.5 (\frac{b}{b_{min}(\rho^2/D)} - \frac{7}{5})}.  \label{eq:xi_def} \ee
\end{lem}
\begin{IEEEproof}
The proof is nearly identical to that of Lemma \ref{lem:good} given in Section \ref{sec:proof_lem}, with the terms $L^{-3/2}$ and $\frac{3}{2L}$ replaced by 
$L^{-5/2}$ and $\frac{5}{2L}$, respectively, throughout the lemma.  Thus we obtain the following condition on $b$ which is the analog of \eqref{eq:bmin_cond0}.
\be
\begin{split}
& b  > \\
&    \max_{\alpha \in \{\frac{1}{L}, \ldots, \frac{L}{L} \}} \hspace{-2pt} \left\{ \frac{R}{(\min\{ \alpha \Lambda(\alpha), c_1 \})}  \left[ \min \left\{ \alpha, \bar{\alpha}, \frac{\log 2} {\log L} \right\}  +  \frac{5 }{ 2L} \right]  \right\} \\
 &  =  \frac{ 3.5 R}{\Lambda(0)} + O\left( \frac{1}{L} \right) \\  
& = \frac{7}{5} b_{min}\left( \frac{\rho^2}{D} \right) + O\left( \frac{1}{L} \right).
\end{split}
\label{eq:bmin_cond0}
\ee
The result is then obtained using arguments analogous to \eqref{eq:pxalph0} and \eqref{eq:pxalph1}.
\end{IEEEproof}

We now compute each of the three terms in the RHS of Suen's inequality.  

\textbf{First Term $\frac{\lambda}{2}$}: We have
\be
\begin{split}
& \lambda = \sum_{i =1}^{e^{nR}} \expec V_i = \expec X_g  \\
&  \stackrel{(a)}{=}  \expec X \cdot P(  \beta(1) \text{ is $\e$-good } \mid \beta(1) \text{ is a solution} ),
\end{split}
\label{eq:lambda_expand}
\ee
where $(a)$ follows from \eqref{eq:pvu}.  Given that $\beta(1)$  is a solution, Lemma \ref{lem:good_strong} shows that
\be
 \sum_{\alpha =  \frac{1}{L}, \ldots, \frac{L}{L}} X_{\alpha}(\beta(1)) < (\expec X )  L^{-3/2}
\label{eq:b1good_strong}
\ee
with probability at least $1 - \xi$. As $\e= L^{-3/2}$,   $\beta(1)$ is $\e$-good according to Definition \ref{def:good}  if \eqref{eq:b1good_strong} is satisfied. Thus  the RHS of \eqref{eq:lambda_expand} can be lower bounded as follows.
\be
\begin{split}
 \lambda & =  \expec X \,  \cdot \,  P(  \beta(1) \text{ is $\e$-good } \mid \beta(1) \text{ is a solution} ) \\
&  \geq   \expec X \, \cdot \, (1 - \xi).
\end{split}
\label{eq:lambda_lb}
\ee
Using the expression from \eqref{eq:EX} for the expected number of solutions $\expec X$, we have
\be
\lambda \geq (1 - \xi) \frac{\kappa}{\sqrt{n}} e^{n (R - \frac{1}{2} \log \frac{\rho^2}{D})},
\label{eq:lambda_lb1}
\ee
where $\kappa >0$ is a  constant. For $b > \frac{7}{5} b_{min}(\rho^2/D)$,  \eqref{eq:xi_def} implies  that $\xi$ approaches $1$ with growing $L$.

\textbf{Second term ${\lambda}/(6\delta)$}: Due to the symmetry of the code construction, we have
\be
\begin{split}
& \delta   =  \max_{i\in \{1, \ldots, e^{nR} \}} \sum_{k \sim i}  P\left(V_k  =1 \right) \\
& = \sum_{k \sim i}  P\left(V_k =1  \right) \quad \forall i \in \{1, \ldots, e^{nR} \} \\
&= \sum_{r=1}^{L-1} {L \choose r} (M-1)^{L-r} \cdot P\left(V_1 =1 \right) \\
& = \left(M^L - 1 - (M-1)^L\right) P\left(V_1 =1 \right).
\end{split}
\ee
Combining this together with the fact that \[ \lambda =  \sum_{i=1}^{M^L} \expec V_i = M^L \, P(V_1 =1), \]
we obtain
\be
\frac{\lambda}{\delta}  = \frac{M^L}{M^L - 1 - (M-1)^L} = \frac{1}{1- L^{-bL} - (1- L^{-b})^L},
\label{eq:lamb_over_del}
\ee
where the second equality is obtained by substituting $M=L^b$. Using a Taylor series bound for the denominator of \eqref{eq:lamb_over_del} (see \cite[Sec. V]{RVGaussianML} for details) yields the following lower bound for sufficiently large $L$:
\be
\frac{\lambda}{\delta} \geq \frac{L^{b-1}}{2}.
\label{eq:lamb_Del_lb}
\ee

\textbf{Third Term $\lambda^2/(8\Delta)$}: We have
\be
\begin{split}
& \Delta  = {\frac{1}{2} \sum_{i=1}^{M^L} \sum_{j \sim i} \expec\left[V_i V_j \right] }  \\
&  =   {\frac{1}{2} \sum_{i=1}^{M^L} P(V_i=1) \sum_{j \sim i}  P(V_j=1 \mid V_i=1)}  \\
& \stackrel{(a)}{=}  \frac{1}{2} \,  \expec X_g \sum_{j \sim 1}  P(V_j=1 \mid V_1=1)  \\
& =   \frac{1}{2} \,  \expec X_g  \,  \expec\Big [ \sum_{j \sim 1} \mathbf{1}\{V_j=1\}  \mid V_1=1\Big ] \\
&  \stackrel{(b)}{\leq} \frac{1}{2} \,  \expec X_g  \, \expec\Bigg [\sum_{\alpha = \frac{1}{L}, \ldots, \frac{L-1}{L}} \hspace{-4pt} X_{\alpha} (\beta(1)) \mid  V_1=1 \Bigg].
\end{split}
\label{eq:Delta_expand}
\ee
In \eqref{eq:Delta_expand}, $(a)$ holds because of the symmetry of the code construction. The inequality $(b)$ is obtained as follows. The number of $\e$-good solutions that share common terms with $\beta(1)$ is bounded above by the total number of solutions sharing common terms with $\beta(1)$.  The latter quantity can be expressed as the sum of the number of solutions sharing  exactly $\alpha L$ common terms with $\beta(1)$, for $\alpha \in \{\tfrac{1}{L}, \ldots, \tfrac{L-1}{L} \}$.

Conditioned on $V_1=1$, i.e., the event that $\beta(1)$ is a $\e$-good solution, the total number of solutions that share common terms with $\beta(1)$ is bounded by $\e \, \expec X$. Therefore, from \eqref{eq:Delta_expand} we have
\be
\begin{split}
& \Delta \leq   \frac{1}{2} \expec X_g  \, \expec\Bigg [\sum_{\alpha = \frac{1}{L}, \ldots, \frac{L-1}{L}} \hspace{-4pt} X_{\alpha} (\beta(1)) \mid  V_1=1 \Bigg]  \\
& \leq \frac{1}{2} \left(\expec X_g \right) (L^{-3/2}\,  \expec X) \leq \frac{L^{-3/2}}{2} (\expec X)^2,
\end{split}
\label{eq:Delta_ub}
\ee
where we have used $\e = L^{-3/2}$, and the fact that $X_g \leq X$. Combining \eqref{eq:Delta_ub} and \eqref{eq:lambda_lb}, we obtain
\be
\frac{\lambda^2}{8 \Delta} \geq \frac{ (1-\xi)^2 (\expec X)^2}{4 L^{-3/2}  (\expec X)^2} \geq \kappa L^{3/2},
\label{eq:lamb_Del2_lb}
\ee
where $\kappa$ is a strictly positive constant.

\textbf{Applying Suen's inequality}:  Using the lower bounds obtained in \eqref{eq:lambda_lb1},  \eqref{eq:lamb_Del_lb}, and  \eqref{eq:lamb_Del2_lb}  in  \eqref{eq:suens_ineq}, we obtain
\be
\begin{split}
& P\left( \sum_{i=1}^{e^{nR}} V_i \right)  \\
& \leq \exp \left( - \kappa \, \min \left\{ e^{n (R - \frac{1}{2} \log \frac{\rho^2}{D} -\frac{\log n}{2 n})}, \, L^{b-1}, \, L^{3/2} \right\} \right),
\end{split}
\label{eq:sum_vi_0}
\ee
where $\kappa$ is a positive constant. Recalling from \eqref{eq:rel_nL} that $L = \Theta( \tfrac{n}{\log n})$ and  $R > \frac{1}{2} \ln \frac{\rho^2}{D}$, we see that for $b >2$,
\be P\left( \sum_{i=1}^{e^{nR}} V_i \right) \leq \exp \left( - \kappa n^{1+c} \right),  \label{eq:sum_vi_bound} \ee
where $c >0$ is a constant. Note that the condition $b > \frac{7}{5} b_{min}(\rho^2/D)$  was also needed to  obtain \eqref{eq:sum_vi_0} via Suen's inequality. In particular, this condition  on $b$ is required for $\xi$ in Lemma \ref{lem:good_strong} to go to $0$ with growing $L$.

Using \eqref{eq:sum_vi_bound} in \eqref{eq:err_bound},  we conclude that for any $\gamma^2 \in (\sigma^2, D^{e^2R})$ the probability of excess distortion can be bounded as  
\be
\begin{split}
P_{e,n} & \leq P(\abs{\bfs}^2 \geq \gamma^2) +  \max_{\rho^2 \in (D, \gamma^2)} P(\mc{E}(\tilde{\bfs}) \mid \abs{\tilde{\bfs}}^2 = \rho^2) \\
& \leq P(\abs{\bfs}^2 \geq \gamma^2)  + \exp(-\kappa n^{1+c}),
\end{split}
\label{eq:Pen_exp_bound}
\ee
provided the parameter $b$ satisfies
\be
b > \max_{\rho^2 \in (D, \gamma^2)} \max \left\{  2, \  \frac{7}{5} b_{min}\left( \rho^2/D \right)  \right \}.
\label{eq:b_bound}
\ee
It can be verified from the definition in \eqref{eq:bmin_def} that $b_{min}(x)$ is strictly increasing in $x \in (1, e^{2R})$.  Therefore, the maximum on the RHS of \eqref{eq:b_bound} is  bounded by $\max \left\{  2, \  \frac{7}{5} b_{min}\left( \gamma^2/D \right)  \right \}$. Choosing $b$ to be larger than this value will guarantee that \eqref{eq:Pen_exp_bound} holds. This completes the proof of the theorem.


\section{Proof  of Lemma \ref{lem:good}} \label{sec:proof_lem}

We begin by listing three useful properties of the function $f(x,y,z)$ defined in \eqref{eq:fdef}. Recall that the probability that an i.i.d. $\mc{N}(0,y)$ sequence is within 
distortion  within distortion $z$ of a norm-$x$ sequence is  $\sim e^{-n f(x,y,z)}$.
\begin{enumerate}
\item For fixed $x,y$, $f$ is strictly decreasing in $z \in (0,x+y)$.
\item For fixed $y,z$,  $f$ is strictly increasing in $x \in (z, \infty)$.
\item For fixed $x,z$ and $x >z$, $f$ is convex in $y$ and attains its minimum value of $\tfrac{1}{2} \log \tfrac{x}{z}$ at $y=x-z$.
\end{enumerate}
These properties are straightforward to verify from the definition \eqref{eq:fdef} using elementary calculus.

For $\mc{K} \subseteq \{1, \ldots, L \}$, let $\beta_\mc{K}$ denote the restriction of $\beta$ to the set $\mc{K}$, i.e., $\beta_{\mc{K}}$ coincides with $\beta$ in the sections indicated by $\mc{K}$ and the remaining entries  are all equal to  zero. For example, if $\mc{K} = \{2, 3\}$,  the second and third sections of $\beta_{\mc{K}}$  will each have one non-zero entry, the other entries are all zeros.

\begin{defi}
Given that $\beta$ is a solution, for $\alpha = \frac{1}{L}, \ldots, \frac{L}{L}$, define $\mc{F}_\alpha(\beta)$ as the event that
\[  \abs{ \tbfs - \mbf{A} \beta_{\mc{K}} }^2 \geq D_{\alpha}  \]
for every size $\alpha L$ subset $\mc{K} \subset \{1, \ldots, L \}$, where $D_\alpha$ 
is the solution to the equation
\be
R \alpha = f(\rho^2, (\rho^2 - D)\alpha, D_{\alpha}).
\label{eq:Dalph}
\ee
\label{def:falph}
\end{defi}
The intuition behind choosing $D_{\alpha}$ according to \eqref{eq:Dalph} is the following. Any subset of $\alpha L$ sections of the design matrix $\mathbf{A}$ defines a  SPARC of rate $R\alpha$, with each codeword consisting of i.i.d $\mc{N}(0, (\rho^2 - D)\alpha)$ entries. (Note that the entries of a single codeword are i.i.d., though the codewords are dependent due to the SPARC structure.) The probability that a codeword from this rate $R\alpha$ code is within distortion $z$ of the source sequence $\tbfs$
is $\sim e^{-n f(\rho^2, (\rho^2 - D)\alpha, z)}$. Hence the expected number of codewords in the rate $R \alpha$ codebook within distortion $z$ of $\tbfs$ is
\[ e^{nR\alpha} e^{-n f(\rho^2, (\rho^2 - D)\alpha, z)}. \]
As $f(\rho^2, (\rho^2 - D)\alpha, z)$ is a strictly decreasing function of $z$ in $(0, \rho^2)$, \eqref{eq:Dalph} says that $D_{\alpha}$ is the smallest expected distortion for any rate $R \alpha$ code with codeword entries chosen i.i.d. $\mc{N}(0, (\rho^2 - D)\alpha)$. \footnote{Note that $D_\alpha$ is \emph{not} the distortion-rate function at rate $R \alpha$ as the codewords are not chosen with the optimal variance for rate $R\alpha$.} For $z < D_\alpha$, the expected number of codewords within distortion $z$ of $\tbfs$ is vanishingly small.

Conditioned on $\mc{F}_\alpha(\beta)$, the idea is that any $\alpha L$ sections of $\beta$ cannot by themselves represent $\tbfs$ with distortion less than $D_{\alpha}$. In other words, in a typical realization of the design matrix, all the sections contribute roughly equal amounts to finding a codeword within $D$ of $\tbfs$. On the other hand, if some $\alpha L$ sections of the SPARC can represent $\tbfs$ with distortion less than $D_\alpha$,  the remaining $\bar{\alpha} L$ sections have  ``less work" to do---this creates a proliferation of solutions that share these $\alpha L$ common sections with $\beta$. Consequently, the total number of solutions is much greater than $\expec X$ for these atypical design matrices.

The first step in  proving the lemma is to show that for any $\beta$,  the event $\mc{F}_\alpha(\beta)$ holds  w.h.p.  The second step  is showing that  when $\mc{F}_\alpha(\beta)$ holds, the expected number of solutions that share {any} common terms  with $\beta$ is small compared to $\expec X$. Indeed, using $\mc{F}_\alpha(\beta)$ we can write
\be
\label{eq:pxa_ub1}
\begin{split}
 & P\left( X_\alpha(\beta)   >  L^{-3/2} \expec X \right)  \\
 & =P\left( \{ X_\alpha(\beta)  >  L^{-3/2} \expec X\},  \ \mc{F}^c_\alpha(\beta) \right)  \\
 & \quad +  P\left( \{ X_\alpha(\beta)  >  L^{-3/2}\, \expec X\}, \ \mc{F}_\alpha(\beta) \right)  \\
& \leq  P(\mc{F}^c_\alpha(\beta)) + P(\mc{F}_\alpha(\beta)) P\left(  X_\alpha(\beta)  <  L^{-3/2} \expec X  \mid \mc{F}_\alpha(\beta) \right)  \\
& \leq   P(\mc{F}^c_\alpha(\beta)) + \frac{\expec[ X_\alpha(\beta) \mid  \mc{F}_\alpha(\beta)] }{L^{-3/2} \, \expec X}
\end{split}
\ee
where  the last line follows from Markov's inequality. We will show that the probability on the left side of \eqref{eq:pxa_ub1} is small  for any solution $\beta$ by showing that each of the two terms on the RHS of \eqref{eq:pxa_ub1} is small. First, a bound on $D_\alpha$.
\begin{lem} For $\alpha \in (0,1]$,
\be R \alpha  >  f( \rho^2, \, (\rho^2 - D)\alpha, \, \rho^2 \bar{\alpha} + D\alpha ) = \frac{1}{2} \log \frac{\rho^2}{ \rho^2 \bar{\alpha} + D\alpha }.
\label{eq:lem1}\ee
Consequently, $D_\alpha < \rho^2 \bar{\alpha} + D\alpha$ for $\alpha = \tfrac{1}{L}, \ldots, \tfrac{L}{L}$.
\label{lem:dalph_ub0}
\end{lem}
\begin{proof}
The last equality in \eqref{eq:lem1} holds because $f(x, x-z, z) = \tfrac{1}{2} \ln \frac{x}{z}$.  Define a function
 \[ g(\alpha) = R \alpha -   \frac{1}{2} \log \frac{\rho^2}{ \rho^2 \bar{\alpha} + D\alpha }. \]
Then $g(0)=0$, $g(1) = R - \frac{1}{2} \ln \frac{\rho^2}{D} >0$, and  the second derivative is
\[ \frac{d^2 g}{d \alpha^2} =  \frac{-(1 - \tfrac{D}{\rho^2})^2}{ (1-(1 - \tfrac{D}{\rho^2}) \alpha)^2 } < 0.\]
Therefore $g$ is strictly concave in $[0,1]$, and its minimum value (attained at $\alpha=0$) is $0$. This proves \eqref{eq:lem1}. Recalling the definition of $D_\alpha$ in \eqref{eq:Dalph}, \eqref{eq:lem1} implies that
\[ f(\rho^2, (\rho^2 - D)\alpha, D_{\alpha}) = R \alpha  >  f( \rho^2, \, (\rho^2 - D)\alpha, \, \rho^2 \bar{\alpha} + D\alpha )   \]
As $f$ is decreasing in its third argument (the distortion),   we conclude that
$D_\alpha < \rho^2 \bar{\alpha} + D\alpha$.
\end{proof}

We now bound each term on the RHS of \eqref{eq:pxa_ub1}.
Showing that the first term of \eqref{eq:pxa_ub1} is small implies that w.h.p any $\alpha L$ sections by themselves will leave a residual distortion of at least $D_{\alpha}$. Showing that the second  term is small implies that under this condition, the expected number of solutions sharing any common terms with $\beta$ is small compared to $\expec X$.

\textbf{ Bounding $\mc{F}^c_\alpha(\beta)$}: From the definition of the event $\mc{F}_\alpha(\beta)$, we have
\be
P(\mc{F}^c_\alpha(\beta)) = \cup_{\mc{K}} \  P(\abs{ \tbfs - \mbf{A} \beta_{\mc{K}} }^2 < D_{\alpha} \mid \beta \text{ is a solution} )
\label{eq:PFc0}
\ee
where the union is over all size-$\alpha L$ subsets of $\{1, \ldots, L\}$. Using a union bound, \eqref{eq:PFc0} becomes
\be
P(\mc{F}^c_\alpha(\beta)) \leq {L \choose {L\alpha} } \frac{P(\abs{ \tbfs - \mbf{A} \beta_{\mc{K}} }^2 < D_{\alpha}, \, \abs{ \tbfs - \mbf{A} \beta }^2 < D)}
{P(\abs{ \tbfs - \mbf{A} \beta}^2 < D)}
\label{eq:PFc1}
\ee
where $\mc{K}$ is a generic size-$\alpha L$ subset of $\{1, \ldots, L\}$, say $\mc{K} = \{1, \ldots, \alpha L\}$.
Recall from \eqref{eq:EX} that for sufficiently large $n$, the denominator in \eqref{eq:PFc1} can be bounded from below as 
\be P(\abs{ \tbfs - \mbf{A} \beta}^2 < D) \geq \frac{\kappa}{\sqrt{n}} e^{-nf(\rho^2, \rho^2 -D, D)} \label{eq:full_cwd} \ee
and $f(\rho^2, \rho^2 -D, D) = \tfrac{1}{2} \log \frac{\rho^2}{D}$.
The numerator in \eqref{eq:PFc1} can be expressed as
\be
\begin{split}
& P(\abs{ \tbfs - \mbf{A} \beta_{\mc{K}} }^2 < D_{\alpha}, \, \abs{ \tbfs - \mbf{A} \beta }^2 < D)  \\
&  = \int_{0}^{D_\alpha} \psi(y) \, P( \abs{ \tbfs - \mbf{A} \beta }^2 < D  \mid \abs{ \tbfs - \mbf{A} \beta_{\mc{K}} }^2 = y) \, dy
\end{split}
\label{eq:PFc_num0}
\ee
where $\psi$ is the density of the random variable $\abs{ \tbfs - \mbf{A} \beta_{\mc{K}} }^2$. Using the cdf at $y$ to  bound $\psi(y)$ in the RHS of \eqref{eq:PFc_num0}, we obtain the following upper bound for sufficiently large $n$.
\be
\begin{split}
& P(\abs{ \tbfs - \mbf{A} \beta_{\mc{K}} }^2 < D_{\alpha}, \, \abs{ \tbfs - \mbf{A} \beta }^2 < D) \\
& \leq  \int_{0}^{D_\alpha} P(\abs{ \tbfs - \mbf{A} \beta_{\mc{K}} }^2 <y)  \\
& \qquad \qquad   \cdot P( \abs{ \tbfs - \mbf{A} \beta }^2 < D \mid \abs{ \tbfs - \mbf{A} \beta_{\mc{K}} }^2 = y) \, dy \\
&  \stackrel{(a)}{\leq} \int_{0}^{D_\alpha} \frac{\kappa}{\sqrt{n}} e^{-n f(\rho^2, (\rho^2-D)\alpha, y)}  \\
& \quad \ \   \cdot P( \abs{ (\tbfs - \mbf{A} \beta_\mck) - \mbf{A} \beta_{\mckc}  }^2 < D \mid \abs{ \tbfs - \mbf{A} \beta_{\mc{K}} }^2 = y) \, dy \\
& \stackrel{(b)}{\leq} \int_{0}^{D_\alpha}  \frac{\kappa}{\sqrt{n}} e^{-n f(\rho^2, (\rho^2-D)\alpha, y)} \cdot e^{-n f(y, (\rho^2-D)\bar{\alpha}, D)}  \,dy \\
& \stackrel{(c)}{\leq} \int_{0}^{D_\alpha} \frac{\kappa}{\sqrt{n}}  e^{-n f(\rho^2, (\rho^2-D)\alpha, D_\alpha)} \cdot
 e^{-n f(D_\alpha, (\rho^2-D)\bar{\alpha}, D)}  \,dy.
\end{split}
\label{eq:PFc_num1}
\ee
In \eqref{eq:PFc_num1},  $(a)$  holds for sufficiently large $n$ and is obtained using the strong version of Cram{\'e}r's large deviation theorem: note that $\mbf{A} \beta_{\mck}$ is   a linear combination of $\alpha L$ columns of $\mbf{A}$, hence it is a Gaussian random vector with i.i.d. $\mc{N}(0, (\rho^2 -D)\alpha)$ entries that is independent of $\tbfs$.    Inequality $(b)$ is similarly obtained: $\mbf{A} \beta_{\mckc}$ has i.i.d.   $\mc{N}(0, (\rho^2 -D)\bar{\alpha})$ entries, and is independent of both $\tbfs$ and $\mbf{A} \beta_{\mck}$.   Finally,  $(c)$ holds because the overall exponent
\[ f(\rho^2, (\rho^2-D)\alpha, y) + f(y, (\rho^2-D)\bar{\alpha}, D)  \]
is a decreasing function of $y$, for $y \in (0, \rho^2 \bar{\alpha} + D \alpha ]$, and $ D_{\alpha} \leq  \rho^2 \bar{\alpha} + D \alpha$.

Using \eqref{eq:full_cwd} and \eqref{eq:PFc_num1} in \eqref{eq:PFc1}, for sufficiently large $n$ we have
\be
\begin{split}
& P(\mc{F}^c_\alpha(\beta)) \leq  \kappa  {L \choose {L\alpha} }   \\
& \quad   \times e^{-n[ f(\rho^2, (\rho^2 - D)\alpha, D_{\alpha}) + f(D_\alpha, (\rho^2-D)\bar{\alpha}, D) -  f(\rho^2, \rho^2 -D, D) ]}. \end{split}
\label{eq:PFc_final}
\ee

\textbf{Bounding $\expec[ X_\alpha(\beta) | \, \mc{F}_\alpha(\beta)]$}: 
There are ${L \choose L\alpha}(M-1)^{L \bar{\alpha}}$ codewords which share  $\alpha L$ common terms with $\beta$. Therefore
\be
\begin{split}
& \expec[ X_\alpha(\beta) \mid  \mc{F}_\alpha(\beta)] = {L \choose L\alpha}(M-1)^{L \bar{\alpha}}  \\
& \qquad  \qquad  \times P(\abs{\tbfs - \mbf{A} \beta' }^2 < D \mid \abs{\tbfs - \mbf{A} \beta }^2 < D, \, \mc{F}_\alpha(\beta) )
\end{split}
\label{eq:EXalph0}
\ee
where $\beta'$ is a  codeword that shares exactly $\alpha L$ common terms with $\beta$. If $\mck$ is the size-$\alpha L$ set of common sections between $\beta$ and $\beta'$, then $\beta'= \beta_{\mck} + \beta'_{\mckc}$ and
\be
\begin{split}
&  P(\abs{\tbfs - \mbf{A} \beta' }^2 < D \mid \abs{\tbfs - \mbf{A} \beta }^2 < D, \, \mc{F}_\alpha(\beta) )  \\
&= P(\abs{(\tbfs - \mbf{A} \beta_{\mck}) - \mbf{A} \beta'_{\mckc} }^2 < D \mid \abs{\tbfs - \mbf{A} \beta }^2 < D, \, \mc{F}_\alpha(\beta) ) \\
& \stackrel{(a)}{\leq} P\left( \tfrac{1}{n} \sum_{i=1}^n (D_\alpha - (\mbf{A} \beta'_{\mckc})_i)^2 \; < D \right) \\
& \stackrel{(b)}{\leq} \frac{\kappa}{\sqrt{n}}  e^{-n f(D_\alpha, (\rho^2-D)\bar{\alpha}, D)},
\end{split}
\label{eq:bbpr}
\ee
where $(b)$ holds for sufficiently large $n$. In \eqref{eq:bbpr}, $(a)$  is obtained as follows.    Under the event $\mc{F}_\alpha(\beta)$, the norm $\abs{\tbfs - \mbf{A} \beta_{\mck}}^2$ is at least $D_\alpha$, and $\mbf{A} \beta'_{\mckc}$ is an i.i.d. $\mc{N}(0, (\rho^2 -D) \bar{\alpha})$ vector independent of $\tbfs$, $\beta$, and $\beta_{\mck}$. $(a)$ then follows from the rotational invariance of the distribution of $\mbf{A} \beta'_{\mckc}$. Inequality $(b)$ is obtained using the strong version of Cram{\'e}r's large deviation theorem.

Using \eqref{eq:bbpr} in \eqref{eq:EXalph0}, we obtain for sufficiently large $n$
\be
\begin{split}
 & \expec[ X_\alpha(\beta) \mid  \mc{F}_\alpha(\beta)]  \\
 &  \leq {L \choose L\alpha}(M-1)^{L \bar{\alpha}} \frac{\kappa}{\sqrt{n}} e^{-n f(D_\alpha, (\rho^2-D)\bar{\alpha}, D)}  \\
 &  \leq {L \choose L\alpha} \frac{\kappa}{\sqrt{n}} e^{n (R\bar{\alpha} - f(D_\alpha, (\rho^2-D)\bar{\alpha}, D))}.
\end{split}
\label{eq:EXalph_bound}
\ee

\textbf{Overall bound}: Substituting the bounds from \eqref{eq:PFc_final}, \eqref{eq:EXalph_bound} and \eqref{eq:EX} in \eqref{eq:pxa_ub1}, for sufficiently large $n$ we have  for $\frac{1}{L} \leq \alpha \leq 1$:
\be
\begin{split}
& P\left( X_\alpha(\beta)  >  L^{-3/2} \expec X \right)  \leq \kappa {L \choose L\alpha}  \\
&  \times \Big( e^{-n[ f(\rho^2, (\rho^2 - D)\alpha, D_{\alpha}) + f(D_\alpha, (\rho^2-D)\bar{\alpha}, D) -  f(\rho^2, \rho^2 -D, D) ]} \\
&+   L^{3/2}  e^{-n [R\alpha + f(D_\alpha, (\rho^2-D)\bar{\alpha}, D) - f(\rho^2, (\rho^2-D), D)]}  \Big).
\end{split}
\label{eq:pxa_ub2}
\ee

Since $D_\alpha$  is chosen to satisfy $R\alpha = f(D_\alpha, (\rho^2-D)\bar{\alpha}, D) $, the two exponents in \eqref{eq:pxa_ub2} are equal. To bound \eqref{eq:pxa_ub2}, we use the following lemma. 
\begin{lem}
For $\alpha \in \{ \frac{1}{L}, \ldots, \frac{L-1}{L} \}$, we have
\be
\begin{split}
& \Big[ f(\rho^2, (\rho^2 - D)\alpha, D_{\alpha}) + f(D_\alpha, (\rho^2-D)\bar{\alpha}, D) \\
& \quad  - f(\rho^2, (\rho^2-D), D) \Big ]  \\
& \qquad  >
\left\{ \begin{array}{l l}  \alpha \Lambda(\alpha) &  \text{ if } D_\alpha > D \\ 
c_1 &\ \text{ if } D_\alpha \leq D. \end{array} \right.
\end{split}
\label{eq:gap_claim}
\ee
where $D_\alpha$  is  the solution of \eqref{eq:Dalph}, $c_1$ is a positive constant given by \eqref{eq:c1_def}, and 
\be
\begin{split}
& \Lambda(\alpha) = \frac{1}{8} \left(\frac{D}{\rho^2}\right)^4 \left(1 + \frac{D}{\rho^2}\right)^2\left(1 -\frac{D}{\rho^2} \right)  \\
&\cdot  \left[ -1 + \left( 1 +  \frac{2 \sqrt{\rho^2/D}}{(\frac{\rho^2}{D} -1)}  \left(R -\frac{1}{2\alpha} \log \frac{\rho^2}{\rho^2 \bar{\alpha} + D \alpha}\right) \right)^{\frac{1}{2}} \right]^2.
\end{split}
\label{eq:c0def}
\ee
\label{lem:gap_lem}
\end{lem}
\begin{IEEEproof} See Appendix \ref{app:gap_lem_proof}. \end{IEEEproof} 

We observe that $\Lambda(\alpha)$ is strictly decreasing for $\alpha \in (0,1]$. This can be seen by using the Taylor expansion of $\log(1 - x)$ for $0<x <1$  to write
\be
\begin{split}
  R -\frac{1}{2\alpha} \ln \frac{\rho^2}{\rho^2 \bar{\alpha} + D \alpha}  &  =  R + \frac{1}{2 \alpha} \log \left(1 - \alpha \left(1-\frac{D}{\rho^2} \right)\right)  \\
&   = R - \frac{1}{2} \sum_{k=1}^\infty \left( 1 - \frac{D}{\rho^2} \right)^k \frac{\alpha^{k-1}}{k}.
\end{split}
\label{eq:Rminus_taylor}
\ee
Since
\[ R > \frac{1}{2} \log \frac{\rho^2}{D} > \frac{1}{2}\left(1 -\frac{D}{\rho^2}\right), \]
\eqref{eq:Rminus_taylor} shows that $\Lambda(\alpha)$ is  strictly positive and strictly decreasing in $\alpha \in (0,1)$ with 
\be
\begin{split}
&\Lambda(0):= \lim_{\alpha \to 0} \Lambda(\alpha)  = \frac{1}{8}  \left(\frac{D}{\rho^2}\right)^4 \left(1 + \frac{D}{\rho^2}\right)^2\left(1 -\frac{D}{\rho^2} \right)  \\
& \qquad  \left[ -1 + \Bigg( 1 +  \frac{2 \sqrt{\rho^2/D}}{(\frac{\rho^2}{D} -1)}  \left(R -\frac{1}{2}\left(1-\frac{D}{\rho^2} \right) \right) \Bigg)^{\frac{1}{2}} \right]^2, \\ 
& \Lambda(1)   = \frac{1}{8} \left(\frac{D}{\rho^2}\right)^4 \left(1 + \frac{D}{\rho^2}\right)^2\left(1 -\frac{D}{\rho^2} \right)  \\
& \qquad \left[ -1 + \left( 1 +  \frac{2 \sqrt{\rho^2/D}}{(\frac{\rho^2}{D} -1)}  \left(R -\frac{1}{2}\log \frac{\rho^2}{D} \right) \right)^{\frac{1}{2}} \right]^2.
\end{split}
\label{eq:Lambda01}
\ee

Substituting \eqref{eq:gap_claim} in \eqref{eq:pxa_ub2}, we have, for $\alpha \in \left\{\frac{1}{L}, \ldots, \frac{L}{L} \right\}$:
\be
\begin{split}
& P\left( X_\alpha(\beta)  >  L^{-3/2} \expec X \right) \\
& <  \kappa  {L \choose L\alpha} {L^{3/2}} \exp(-n \cdot \min\{ \alpha \Lambda(\alpha), c_1 \} ).
\end{split}
\ee
Taking logarithms and dividing both sides by $L \log L$, we obtain
\be
\begin{split}
& \frac{1}{L \log L} \log P\left( X_\alpha(\beta)  >  L^{-3/2} \expec X \right)   \\
& <  \frac{\log { \kappa}} {L \log L} + \frac{ \log {L \choose L\alpha}}{ L \log L}  + \frac{3}{2L } - \frac{n (\min\{ \alpha \Lambda(\alpha), c_1 \})}{L \log L}\\
& \stackrel{(a)}{=} \frac{\log { \kappa}} {L \log L}  + \min \left\{ \alpha, \bar{\alpha}, \frac{\log 2} {\log L} \right\} +  \frac{3}{2L} \\
& \qquad - \frac{(\min\{ \alpha \Lambda(\alpha), c_1 \}) b}{R}
\end{split}
\label{eq:div_by_llogl}
\ee
where to obtain $(a)$, we have used the bound
\[ \log {L \choose L\alpha} <   \min \ \{ \alpha L \log L, \  (1-\alpha) L \log L, \ L \log 2 \} \]
and the relation \eqref{eq:rel_nL}. For the right side of \eqref{eq:div_by_llogl} to be negative for sufficiently large $L$, we need
\be
\frac{(\min\{ \alpha \Lambda(\alpha), c_1 \}) b}{R} >  \min \left\{ \alpha, \bar{\alpha}, \frac{\log 2} {\log L} \right\} +  \frac{3}{2 L} .
\label{eq:alph_bcond}
\ee
This can be arranged by choosing $b$ to be large enough. Since \eqref{eq:alph_bcond} has to be satisfied for all $\alpha \in \{ \tfrac{1}{L}, \ldots, \tfrac{L-1}{L}\}$, we need
\be
\begin{split}
& b  > \\
&   \max_{\alpha \in \{\frac{1}{L}, \ldots, \frac{L}{L} \}} \hspace{-2pt}\left\{ \frac{R}{(\min\{ \alpha  \Lambda(\alpha), c_1 \})}  \Big[ \min 
\Big\{ \alpha, \bar{\alpha}, \frac{\log 2} {\log L} 
\Big\}  +  \frac{3 }{ 2L} \Big]  \right\} \\
&  \stackrel{(a)}{=} \frac{ 2.5 R}{ \Lambda(0)} +   O\left( \frac{1}{L} \right) \\
& = b_{min}\left( \frac{\rho^2}{D} \right) +  O\left( \frac{1}{L}\right).
\end{split}
\label{eq:bmin_cond0}
\ee
In \eqref{eq:bmin_cond0}, $(a)$ holds because $\Lambda(\alpha)$ is of constant order for all $\alpha \in (0,1]$, hence the maximum is attained at $\alpha = \tfrac{1}{L}$. The constant $\Lambda(0)$ is given by \eqref{eq:Lambda01}, and $b_{min}(.)$ is defined in the statement of Theorem
\ref{thm:ml_result}.

When $b$ satisfies \eqref{eq:bmin_cond0} and $L$ is sufficiently large, for $\alpha \in \{ \tfrac{1}{L}, \ldots, \tfrac{L}{L} \}$,  the bound in \eqref{eq:div_by_llogl} becomes
\be
\begin{split}
& \frac{1}{L \log L} \log  P\left( X_\alpha(\beta) >  L^{-3/2} \expec X \right)  \\
 & < \frac{\log {\kappa }} {L \log L}   -   \frac{\min\{ \alpha \Lambda(\alpha), c_1 \} (b - b_{min} - O(\frac{1}{L}))}{R}  \\
&  \leq    \frac{\log {\kappa}} {L \log L}   - \frac{\Lambda(0)}{L} \frac{ (b- b_{min})}{R}  =  \frac{\log {\kappa }} {L \log L}   - \frac{2.5 (\frac{b}{b_{min} }- 1)}{L}.
\end{split}
\label{eq:pxalph0}
\ee
Therefore
\be  P\left( X_\alpha(\beta) >   L^{-3/2} \expec X \right) <  \, \kappa L^{ -2.5 (\frac{b}{b_{min}} - 1)}. \label{eq:pxalph1} \ee
This completes the proof of Lemma \ref{lem:good}.

\appendices

\section{Proof of Lemma \ref{lem:gap_lem}} \label{app:gap_lem_proof}

For $\alpha \in \{\tfrac{1}{L}, \ldots, \tfrac{L-1}{L} \}$, define the function $g_\alpha: \mathbb{R} \to \mathbb{R}$ as
\be
g_\alpha(u) = f(\rho^2, (\rho^2-D)\alpha, u) + f(u,  (\rho^2-D)\bar{\alpha}, D)
 - \frac{1}{2} \ln \frac{\rho^2}{D}. \ee
 We want a lower bound for $g_\alpha(D_\alpha)  \geq \Lambda(\alpha) \alpha$, where $D_\alpha$ is the solution to
 \be R \alpha = f(\rho^2, (\rho^2-D)\alpha, D_\alpha). \label{eq:Dalph_def} \ee
We consider the cases $D_\alpha > D$ and $D_\alpha \leq D$ separately.  Recall from Lemma \ref{lem:dalph_ub0} that $D_\alpha < \rho^2 \bar{\alpha} + D{\alpha}$. 
 
 \subsubsection*{Case $1$:  $D< D_\alpha < \rdabara$.} 
 
 In this case, both the $f(.)$ terms in the definition of $g_{\alpha}(D_\alpha)$ are  strictly positive. We can write
 \be D_\alpha =  \rdabara - \delta, \label{eq:delta_def0}\ee
where $\delta \in (0, (\rho^2-D)\bar{\alpha})$. Expanding $g_\alpha(\rdabara - \delta)$  around $\rdabara$ using Taylor's theorem, we obtain
\be
g(D_\alpha) = g(\rdabara) - g^\prime(\rdabara) \delta + g^{\prime \prime}(\xi) \frac{\delta^2}{2} =  g^{\prime \prime}(\xi) \frac{\delta^2}{2}
\label{eq:galph_taylor}
\ee 
since $g(\rdabara) = g^\prime(\rdabara) =0$. Here $\xi$ is a number in the interval $(D, \rdabara)$. We bound $g(D_\alpha)$ from below  by obtaining separate lower bounds for 
$g^{\prime \prime}(\xi)$ and $\delta$.

\emph{Lower Bound for $g^{\prime \prime}(\xi)$}: Using the definition of $f$ in \eqref{eq:fdef}, the second derivative of $g(u)$ is 
\be
\begin{split}
& g^{\prime \prime}(u)  =  \frac{-1}{2 u^2}  \\
& + \frac{2 \rho^4 \, \cdot [ (\rho^2-D)^2 \alpha^2 + 4\rho^2 u]^{-1/2}}{ \alpha (\rho^2-D) \left[\sqrt{(\rho^2-D)^2 \alpha^2 + 4\rho^2 u} - \rda\right]^2}  \\
& + \frac{2D^2 \, \cdot  [(\rho^2-D)^2 \bar{\alpha}^2 + 4D u]^{-1/2} }{\bar{\alpha} (\rho^2-D) \left[\sqrt{(\rho^2-D)^2 \bar{\alpha}^2 + 4D u} - \rdabar \right]^2}.
\end{split}
\ee
It can be verified that $g^{\prime \prime}(u)$ is a decreasing function, and hence for $\xi \in (D, \rdabara)$,
\be
\begin{split}
&g^{\prime \prime}(\xi) \geq g^{\prime \prime}(\rdabara)\\
 &=  \frac{-1}{2(\rdabara)^2} \\
 & \quad  + \frac{\rho^4}{2 \alpha (\rho^2-D) (\rdabara)^2 (\rho^2(1+\bar{\alpha})+D \alpha)} \\
&\quad  + \frac{1}{2 \bar{\alpha} (\rho^2-D)(\rho^2 \bar{\alpha} + D(1+\alpha))}\\
 & =  \frac{(\rho^2 +D)}{2\alpha \bar{\alpha} (\rho^2-D) (\rho^2(1+\bar{\alpha})+D \alpha)(\rho^2 \bar{\alpha} + D(1+\alpha))} \\
 & \geq 
 \frac{1}{4 \alpha \bar{\alpha} (\rho^2-D) \rho^2}.
\end{split}
\label{eq:gxi_bnd}
\ee

\emph{Lower bound for $\delta$}: From \eqref{eq:Dalph_def} and \eqref{eq:delta_def0},  note that $\delta$ is the solution to 
\be
R \alpha = f(\rho^2, (\rho^2-D)\alpha, \rdabara - \delta).
\ee
Using Taylor's theorem for  $f$ in its third argument around the point $p :=(\rho^2, (\rho^2-D)\alpha, \rdabara)$, we have
\be
\begin{split}
& R \alpha  = f(p) - \frac{\partial f}{\partial z} \Big{\lvert}_p  \delta + \frac{\partial^2  f}{\partial z^2} \Big{\lvert}_{\tilde{p}}  \, \frac{\delta^2}{2}  \\
& = \frac{1}{2} \ln \frac{\rho^2}{\rho^2 \bar{\alpha} + D \alpha} + \frac{1}{2( \rdabara)} \, \delta +   \frac{1}{2} \frac{\partial^2  f}{\partial z^2} \Big{\lvert}_{\tilde{p}}  \delta^2,
\end{split}
\label{eq:delta_quad}
\ee
where $\tilde{p}=p=(\rho^2, (\rho^2-D)\alpha, \tilde{z})$ for some $\tilde{z} \in (D, \rdabara)$.  As $\eqref{eq:delta_quad}$ is a quadratic in $\delta$ with positive coefficients for the $\delta$ and $\delta^2$ terms, replacing the $\delta^2$ coefficient with an upper bound and solving the resulting quadratic will yield a \emph{lower bound} for $\delta$. Since the function
\be
\frac{\partial^2 f}{\partial z^2}\Big{\lvert}_{(x,y,z)} = \frac{ 2x^2}{y \sqrt{y^2 + 4xz}(  \sqrt{y^2 + 4xz} - y)^2}
\ee
is decreasing in $z$, the $\delta^2$ coefficient can be bounded as follows.
\be
\frac{1}{2} \frac{\partial^2 f}{\partial z^2}\Big{\lvert}_{\tilde{p} =(\rho^2, (\rho^2-D)\alpha, \tilde{z})} 
   \leq a^*   := \frac{1}{2} \frac{\partial^2 f}{\partial z^2}\Big{\lvert}_{(\rho^2, (\rho^2-D)\alpha, D)}, 
  \ee
  where $a^*$  can be computed to be
  \be
  \begin{split} 
& a^* = \rho^4 \Bigg( \alpha (\rho^2-D) \sqrt{(\rho^2-D)^2 \alpha^2 + 4\rho^2D} \\
& \qquad \quad  \left[\sqrt{(\rho^2-D)^2 \alpha^2 + 4\rho^2D} - \rda\right]^2 \Bigg)^{-1}.
\end{split}
\label{eq:astar_def}
\ee 
Therefore we can obtain a lower bound for $\delta$, denoted by $\underline{\delta}$,  by solving the equation
\be
\underline{\delta}^2 \, a^*  + \underline{\delta} \, \frac{1}{2( \rdabara)} - \left( R \alpha - \frac{1}{2} \ln \frac{\rho^2}{\rho^2 \bar{\alpha} + D \alpha} \right) = 0.
\ee
We thus obtain
\be
\begin{split}
& \delta > \un{\delta} = \frac{1}{4(\rdabara) a^*} \Bigg[ -1 \,  +  \\
& \ \left( 1 + 16(\rdabara)^2a^* \alpha \left(R -\frac{1}{2\alpha} \ln \frac{\rho^2}{\rho^2 \bar{\alpha} + D \alpha} \right)\right)^{1/2}\Bigg].
\end{split}
\label{eq:delta_lb0}
\ee
We now show that $\un{\delta}$ can be bounded from below by $\alpha \Lambda(\alpha)$ by obtaining lower and upper  bounds for $a^* \alpha$. From  \eqref{eq:astar_def} we have 
\be
\begin{split}
& a^* \alpha    = \frac{\rho^4  \cdot [(\rho^2-D)^2 \alpha^2 + 4\rho^2D]^{-1/2}} {(\rho^2-D) \left[\sqrt{(\rho^2-D)^2 \alpha^2 + 4\rho^2D} - \rda\right]^2} \\ 
& \geq  \frac{\sqrt{\rho^2/D}}{8D(\rho^2-D) },
\end{split}
\label{eq:intbnd2}
\ee
where the inequality is obtained by noting that $a ^* \alpha$ is strictly increasing  in $\alpha$, and hence taking $\alpha=0$ gives a lower bound. Analogously, taking $\alpha =1$ yields the upper bound
\be
a^* \alpha \leq \frac{\rho^4}{4 D^2 (\rho^4 - D^2)}. 
\label{eq:intbnd3}
\ee
Using the bounds of \eqref{eq:intbnd2} and \eqref{eq:intbnd3} in \eqref{eq:delta_lb0}, we obtain
\be
\begin{split}
 & \delta > \un{\delta}   \geq \alpha \,  \frac{D^2 ( \rho^4 - D^2)}{\rho^6}  \Bigg[ -1 \  + \\
 & \quad  \left( 1 +  \frac{2 \sqrt{\rho^2/D}}{(\frac{\rho^2}{D} -1)}
 \left(R -\frac{1}{2\alpha} \ln \frac{\rho^2}{\rho^2 \bar{\alpha} + D \alpha} \right)\right)^{1/2}\Bigg].
\end{split}
\label{eq:del_bd_final}
\ee
Finally, using the lower bounds for $g^{\prime\prime}(\xi)$  and  $\delta$ from \eqref{eq:del_bd_final} and \eqref{eq:gxi_bnd} in  \eqref{eq:galph_taylor}, we obtain
\be
\begin{split}
& g(D_\alpha) >  \frac{\alpha}{8} \left(\frac{D}{\rho^2}\right)^4 \left(1 + \frac{D}{\rho^2}\right)^2\left(1 -\frac{D}{\rho^2} \right) \\
& \ \times \left[ -1 + \left( 1 +  \frac{2 \sqrt{\rho^2/D}}{(\frac{\rho^2}{D} -1)} \left(R -\frac{1}{2\alpha} \ln \frac{\rho^2}{\rho^2 \bar{\alpha} + D \alpha} \right)\right)^{\frac{1}{2}}\right]^2 
\\
& = \alpha \Lambda(\alpha).
\end{split}
\ee

 \subsubsection*{Case $2$:  $ D_\alpha \leq D$.}
 In this case, $g(D_\alpha)$ is  given by
 \be
 \begin{split}
&  g(D_\alpha)  = f(\rho^2, (\rho^2-D)\alpha, D_\alpha) + f(D_\alpha,  (\rho^2-D)\bar{\alpha}, D)  \\
& \qquad \qquad - \frac{1}{2} \ln \frac{\rho^2}{D} \\
& =  R\alpha  - \frac{1}{2} \ln \frac{\rho^2}{D},
 \label{eq:g_toprove}
 \end{split}
 \ee
where we have used \eqref{eq:Dalph_def} and the fact that $f(D_\alpha,  (\rho^2-D)\bar{\alpha}, D)=0$ for $D_\alpha \leq D$.
The right hand side of the equation
\[ R \alpha = f (\rho^2, (\rho^2-D)\alpha, z) \]
is decreasing in $z$ for $z \in (0,D]$. Therefore, it is sufficient to consider $D_\alpha =D$ in order to obtain a  lower bound for $R\alpha$ that holds for all $D_\alpha \in (0,D]$.
 
 Next, we claim that the $\alpha$ that solves the equation
 \be
 R \alpha = f (\rho^2, (\rho^2-D)\alpha, D)
 \label{eq:Ralph}
 \ee
lies in the interval $(\tfrac{1}{2},1)$. Indeed, observe that the LHS of \eqref{eq:Ralph} is increasing in $\alpha$, while the RHS is decreasing in $\alpha$ for $\alpha \in (0,1]$. Since the LHS is strictly greater than the RHS at $\alpha=1$ ($R > \frac{1}{2} \ln \frac{\rho^2}{D}$),  the solution is strictly less than $1$. On the other hand, for $\alpha \leq \tfrac{1}{2}$, we have
\be
\begin{split}
R \alpha \leq \frac{R}{2} \leq  \frac{1}{2}\left(1 - \frac{D}{\rho^2} \right)  < \frac{1}{2} \ln \frac{\rho^2}{D} & = f(\rho^2, (\rho^2-D), D) \\
& <  f(\rho^2, \tfrac{(\rho^2-D)}{2}, D),
\end{split}
\ee
i.e., the LHS of \eqref{eq:Ralph} is strictly less than the RHS. Therefore the $\alpha$ that solves \eqref{eq:Ralph} lies in  $(\tfrac{1}{2},1)$. 

To obtain a  lower bound on the RHS of \eqref{eq:g_toprove}, we expand $ f(\rho^2, (\rho^2-D)\alpha, D)$ using Taylor's theorem for the second argument.
\be
\begin{split}
& f (\rho^2, (\rho^2-D)\alpha, D)  =  f (\rho^2, (\rho^2-D) - \Delta, D) \\
 & = \frac{1}{2} \ln \frac{\rho^2}{D} - \Delta \frac{\partial f }{\partial y} \Big{\lvert}_{(\rho^2, (\rho^2-D), D)}  +  
 \frac{\Delta^2}{2} \frac{\partial^2 f }{\partial y^2}\Big{\lvert}_{(\rho^2, y_0, D)} \\
 &  =  \frac{1}{2} \ln \frac{\rho^2}{D}   +  
 \frac{\Delta^2}{2} \frac{\partial^2 f }{\partial y^2}\Big{\lvert}_{(\rho^2, y_0, D)},
 \end{split}
 \label{eq:ftay0}
\ee
where  $\Delta = \rdabar$,  and $y_0$ lies in the interval  $(\tfrac{1}{2} (\rho^2 -D), (\rho^2-D))$. Using \eqref{eq:ftay0} and the shorthand 
\[  f''(y_0) : =  \frac{\Delta^2}{2} \frac{\partial^2 f }{\partial y^2}\Big{\lvert}_{(\rho^2, y_0, D)}, \] 
\eqref{eq:Ralph}  can be written as 
\be
R\alpha - \frac{1}{2} \ln \frac{\rho^2}{D} =  \bar{\alpha}^2 \, \frac{(\rho^2-D)^2}{2}f''(y_0),
\label{eq:rabar_bnd1}
\ee
or
\be
R - \frac{1}{2} \ln \frac{\rho^2}{D} = R \bar{\alpha} \ +  \ \bar{\alpha}^2 \, \frac{(\rho^2-D)^2}{2} f''(y_0). 
\label{eq:abar_quad}
\ee
Solving the quadratic in $\bar{\alpha}$, we get
\be
\bar{\alpha} = \frac{- R + [R^2 + 2(\rho^2-D)^2(R - \frac{1}{2} \ln \frac{\rho^2}{D}) f''(y_0)]^{1/2}}{(\rho^2-D)^2 f''(y_0)}.
\ee
Using this in \eqref{eq:rabar_bnd1}, we get 
\be
\begin{split}
& R\alpha - \frac{1}{2} \ln \frac{\rho^2}{D}  \\
& =  
\frac{\left(- R + \left[R^2 +2 (\rho^2-D)^2 (R - \frac{1}{2} \ln \frac{\rho^2}{D}) f''(y_0)\right]^{1/2}\right)^2}{2(\rho^2-D)^2 f''(y_0)}.
\end{split}
\label{eq:lb_eqn1}
\ee
The LHS is exactly the quantity we want to bound from below. 
From the definition of $f$ in \eqref{eq:fdef}, the second  partial derivative with respect to $y$ can be computed:
\be
f''(y) = \frac{\partial^2 f }{\partial y^2}\Big{\lvert}_{(\rho^2, y, D)} = \frac{\rho^2}{y^3}  + \frac{1}{y^2} - \frac{y}{2(y^2 + 4\rho^2 D)^{3/2}}.
\label{eq:f_dpr}
\ee 
The RHS of \eqref{eq:f_dpr} is strictly decreasing in $y$. We can therefore bound  $f''(y_0)$ as 
\be
\begin{split}
\frac{\rho^2}{(\rho^2-D)^3} < f''(\rho^2-D) < f''(y_0) &  < f''\left(\frac{\rho^2-D}{2}\right)  \\
& < \frac{12 \rho^2}{(\rho^2-D)^3}.
\end{split}
\ee
Substituting these bounds in \eqref{eq:lb_eqn1}, we conclude that for $D_\alpha \leq D$,
\be
\begin{split}
&g(D_\alpha) = R\alpha - \frac{1}{2} \ln \frac{\rho^2}{D}  \\
&  \geq  \, c_1:= \frac{(\rho^2-D)}{24 \rho^2}\left( -R + \left[ R^2  + \frac{2\rho^2 (R - \frac{1}{2} \ln \frac{\rho^2}{D})}{(\rho^2-D)} \right]^{\frac{1}{2}}\right)^2.
\end{split}
\label{eq:c1_def}
\ee

\section*{Acknowledgement}
We thank the anonymous referee for comments which helped improve the paper.


\end{document}